\documentclass[10pt,conference]{IEEEtran}

\usepackage{array}
\usepackage{textcomp}
\usepackage{stfloats}
\usepackage{url}
\usepackage{verbatim}
\usepackage{graphicx}
\usepackage{cite}
\usepackage{longtable}
\usepackage{xcolor}
\usepackage{booktabs}
\usepackage{pifont}
\usepackage{threeparttable}
\usepackage{tabularx}
\usepackage{array}

\usepackage[caption=false,font=footnotesize]{subfig}

\usepackage{amsmath,amssymb,amsfonts}
\usepackage{amsthm}
\newtheorem{theorem}{Theorem}
\newtheorem{lemma}{Lemma}

\theoremstyle{definition}

\usepackage[ruled,vlined,linesnumbered]{algorithm2e}

\makeatletter
\let\@afterindenttrue\@afterindentfalse
\@afterindentfalse
\makeatother
\IEEEoverridecommandlockouts

\begin{document}
\bstctlcite{IEEEexample:BSTcontrol}

\title{STAR-GS: Truthful and Visibility-Aware Resource Scheduling for Ground Station as a Service}
    \author{
        \IEEEauthorblockN{
            Zhiying Wang$^*$, Xiaojian Wang$^\dagger$, Huayue Gu$^\ddagger$,
            Zhishan Guo$^*$, Ruozhou Yu$^*$
        }
        \IEEEauthorblockA{
            $^*$Department of Computer Science, North Carolina State University,
            Raleigh, NC 27606, USA
            \\
            $^\dagger$Department of Computer Science and Engineering,
            University of Colorado Denver, Denver, CO 80204, USA
            \\
            $^\ddagger$Department of Information Technology,
            Kennesaw State University, Kennesaw, GA 30144, USA
        }
    
    \thanks{
    Zhiying Wang, Zhishan Guo, and Ruozhou Yu (\{zwang266, zguo32, ryu5\}@ncsu.edu) are with the
    Department of Computer Science, North Carolina State University, Raleigh, NC 27606, USA.
    Xiaojian Wang (\mbox{xiaojian.wang@ucdenver.edu}) is with the Department of Computer Science and Engineering, University of Colorado Denver, Denver, CO 80204, USA. Huayue Gu (hgu2@kennesaw.edu) is with the Department of Information Technology, Kennesaw State University, Kennesaw, GA 30144, USA. This work was supported in part by the National Science Foundation under grants 2433966, 2414523, 2246672, and 2521121.  The information reported herein does not reflect the position or the policy of the funding agencies.
    }
    }

\maketitle

\begin{abstract}
The rapid growth of Low Earth Orbit satellite constellations has created increasing demand for efficient and scalable downlink services. Ground Station as a Service (GSaaS) provides an on-demand access model for satellite operators, but commercial GSaaS providers must schedule limited ground-station bandwidth among multiple satellites with heterogeneous data demands, overlapping visibility windows, strict deadlines, and strategic bidding behaviors. This paper studies GSaaS resource scheduling from a ground-station-centric perspective, where the provider jointly determines task admission, ground-station assignment, bandwidth allocation, and payments. Under satellite orbital dynamics, bandwidth constraints, and downlink task deadlines, maximizing the provider's revenue is NP-hard. To address this challenge, we propose STAR-GS, a truthful and feasibility-aware scheduling mechanism that combines bid-aware admission control, best-fit ground-station assignment, Earliest Deadline First (EDF)-based bandwidth scheduling, and critical-payment pricing. By integrating auction theory with schedulability analysis, STAR-GS incentivizes task owners to truthfully report their private valuations while ensuring that admitted tasks can be feasibly completed before their deadlines. Simulations using Ansys Systems Tool Kit (STK) show that STAR-GS consistently achieves higher revenue than heuristic baselines, obtains near-MILP performance with substantially lower runtime, and scales smoothly to workloads containing up to 900 tasks.
\end{abstract}


\section{Introduction}
\label{sec:introduction}

The rapid growth in the number of Low Earth Orbit (LEO) satellites, together with applications such as remote sensing, IoT, and meteorological observation \cite{10945753,9442378,9840374}, has created substantial demand for downlink data transmission. Despite that, satellite-ground downlink remains a critical bottleneck, constrained by limited satellite-ground visibility and scarce ground station availability for most LEO constellations. Except for designated broadband LEO constellations such as Starlink or Kuiper \cite{osoro2021techno}, building and operating dedicated ground stations is costly, and each ground station offers only limited coverage. To address the gap, commercial Ground Station as a Service (GSaaS) providers emerge \cite{11044563}, each operating a shared ground stations infrastructure with wide coverage, and leasing their ground station resources to satellite operators on demand.

As LEO satellites move rapidly, their downlink demands vary significantly across locations and time. This spatio-temporal imbalance in downlink demand can lead to congestion at hotspot locations while leaving ground-station resources underutilized elsewhere~\cite{929853,10520815}. Effective satellite-ground downlink resource scheduling is crucial for improving resource utilization while satisfying the demands of diverse LEO applications. To optimize downlink resources in GSaaS, existing works adopt a satellite-centric perspective, where each satellite acts as the decision maker and selects one or multiple ground stations to minimize downlink latency, maximize throughput, or reduce operational cost~\cite{11044527,11044563}. In this scenario, ground stations are modeled as passive resource providers with predefined capacities, while satellites compete or coordinate for the use of ground resources on their own. The model works well when the downlink load is mild and/or static~\cite{Duncan}. However, when ground stations are shared and capacity-constrained, serving multiple satellites with heterogeneous demands and overlapped visibility windows, existing satellite-centric solutions fail to satisfy application requirements and resource utilization~\cite{s25247538}. 

In this paper, we study revenue-maximizing GSaaS scheduling from a ground-station provider's perspective. The problem presents three key challenges: short and overlapping satellite--ground contacts complicate deadline satisfaction~\cite{11329400}; private task valuations require incentive-compatible admission and pricing; and task admission, ground-station assignment, and bandwidth scheduling are tightly coupled. Together, these challenges make revenue maximization non-trivial.

To address these intricate challenges, this paper makes the following major contributions:

\begin{itemize}
    \item We study GSaaS resource scheduling under dynamic satellite--ground visibility, task-specific deadlines, and provider reserve-price constraints, where the provider jointly determines task admission, ground-station assignment, bandwidth allocation, and payment. We formulate this setting as a GSaaS resource scheduling and revenue maximization problem, and prove that the problem is NP-hard.

    \item We propose Scheduling via Truthful Auction and Resource-aware Ground-Station Selection (STAR-GS), a computationally efficient, feasibility-aware scheduling and truthful auction mechanism that combines bid-aware admission, best-fit ground-station assignment, Earliest Deadline First (EDF) scheduling, and critical payment pricing.

    \item We theoretically prove that STAR-GS ensures the completion of admitted tasks through feasibility-aware scheduling, guarantees truthful reporting of private task values, satisfies individual rationality for task owners, and enforces provider-side reserve-price protection.
    
    \item With extensive simulations, we show that STAR-GS achieves near-optimal revenue, outperforms state-of-the-art baselines, and maintains low time complexity and high scalability.
\end{itemize}

\noindent \textbf{Organization.} \S\ref{sec:related_work} reviews related work on GSaaS resource scheduling and auction-based resource allocation. \S\ref{sec:system_model} presents the system model and formulates the GSaaS revenue maximization problem. \S\ref{sec:alg_design} introduces the STAR-GS mechanism and analyzes its theoretical properties. \S\ref{sec:simulation} evaluates the performance of STAR-GS. \S\ref{sec:conclusion} concludes this paper.

\section{Related Work}
\label{sec:related_work}

Research on GSaaS and satellite--ground networking has studied service architecture, decentralized ground-station sharing, satellite-centric ground-station selection, and learning-based downlink optimization. Existing works have considered sustainable offloading through commercial ground stations~\cite{11044527}, blockchain-based GS federation for privacy-preserving nanosatellite communication~\cite{10154485}, cybersecurity analysis of GSaaS architectures~\cite{10115903}, and reinforcement-learning-based downloading under atmospheric uncertainty~\cite{10001260}. Recent studies further optimize multi-provider GSaaS usage and satellite-centric site selection, such as SkyGS~\cite{10858567} and scalable GSaaS site clustering~\cite{kim2025scalablegroundstationselection}. These studies provide effective solutions for scenarios where satellites select suitable ground-station services to improve latency, reliability, cost, or coverage. In contrast, our work focuses on the complementary provider-side setting, where a GSaaS operator manages shared ground-station resources and jointly decides task admission, ground-station assignment, bandwidth scheduling, and pricing.

Resource scheduling for satellite--ground station networks has also been widely investigated. Existing studies have designed constellation-wide scheduling pipelines to mitigate ground-track congestion~\cite{10.1145/3680207.3765249}, reinforcement-learning-based hierarchical scheduling to improve utilization and profit~\cite{CHENG2026123200}, and mixed-integer or genetic-optimization models for satellite--ground network planning~\cite{REN2025126303}. These approaches are well suited to improving scheduling feasibility, throughput, resource utilization, or planning efficiency when task requirements and priorities are specified by the system. Our work considers a different commercial GSaaS setting, where task owners may have private valuations and the provider needs to make admission and pricing decisions while still satisfying visibility and deadline constraints.

Auction mechanisms have been widely used for allocating scarce resources with private valuations. Existing studies have designed pricing and auction mechanisms for edge computing and communication-resource allocation, achieving properties such as truthfulness, budget balance, and efficient resource utilization~\cite{10979921},~\cite{10887300}. Other game-theoretic or Bayesian approaches further optimize surplus and resource efficiency in terrestrial networks. These mechanisms provide useful tools for incentive-aware resource allocation in many networked systems. Compared with these settings, GSaaS scheduling introduces additional structure from intermittent satellite--ground visibility, task-specific deadlines, and ground-station bandwidth constraints. Therefore, our work adapts the mechanism-design perspective to a visibility-aware and deadline-aware GSaaS scheduling problem.

Building on these lines of work, STAR-GS studies provider-side GSaaS resource scheduling by jointly considering revenue maximization, satellite--ground visibility, task deadlines, bandwidth feasibility, and private task valuations. Unlike satellite-centric selection, conventional scheduling, and general auction mechanisms, STAR-GS integrates provider-side admission, assignment, bandwidth scheduling, bid-aware pricing, and visibility- and deadline-aware schedulability checking into a revenue-aware, feasibility-aware, and incentive-compatible mechanism.

\section{System Model and Problem Formulation}
\label{sec:system_model}

\begin{table}[t]
\centering
\caption{Main Notations}
\label{tab:notations}
\footnotesize
\setlength{\tabcolsep}{1.5pt}
\renewcommand{\arraystretch}{0.88}
\begin{tabularx}{\columnwidth}{
>{\raggedright\arraybackslash}p{0.13\columnwidth}
>{\raggedright\arraybackslash}X
>{\raggedright\arraybackslash}p{0.13\columnwidth}
>{\raggedright\arraybackslash}X}
\hline
\textbf{Notation} & \textbf{Description} & \textbf{Notation} & \textbf{Description} \\
\hline
$A_{ik}$ & Arrival slot
& $\mathcal{A}_j(t)$ & Serviceable tasks at GS $j$ \\

$B_j$ & GS bandwidth
& $b_{ijkt}$ & Allocated bandwidth \\

$\mathrm{bid}_{ik}$ & Unit bid
& $c_{ijk}$ & Processing requirement \\

$\mathcal{C}_{ij}$ & Visibility windows
& $D_{ik}$ & Task deadline \\

$\mathcal{F}_{ik}$ & Candidate GS set
& $\mathcal{G}$ & GS set \\

$g(i,k)$ & Assigned GS
& $I_j(t)$ & Unused slot fraction \\

$\mathcal{K}_i$ & Tasks of satellite $i$
& $M_{ij}$ & Number of windows \\

$p_{ik}$ & Unit price
& $P_{ik}$ & Total payment \\

$\mathcal{R}$ & Request set
& $\mathcal{R}^{+}$ & Reserve-eligible requests \\

$\mathcal{S}$ & Satellite set
& $\mathrm{Slack}_j$ & Residual slack \\

$\mathcal{T}$ & Time-slot set
& $\mathcal{T}_{ijk}$ & Feasible service window \\

$u^{\mathrm{GS}}_{ik}$ & Provider utility
& $u^{\mathrm{task}}_{ik}$ & Task-owner utility \\

$v_{ik}$ & Private value
& $\mathcal{W}_j$ & Tasks assigned to GS $j$ \\

$x_{ijk}$ & Assignment indicator
& $y_{ik}$ & Completion indicator \\

$\alpha_{ijkt}$ & Service fraction
& $\Delta t$ & Slot duration \\

$\Delta\tau_{ijkt}$ & Transmitted data
& $\eta_{ij}$ & Spectral efficiency \\

$\rho$ & Reserve price
& $\tau_{ik}$ & Data demand \\

$\tau^{\mathrm{rem}}_{ik}(t)$ & Remaining data
& $\mathcal{W}$ & Candidate/winner set \\
\hline
\end{tabularx}
\end{table}

In this section, we study GSaaS resource scheduling from the perspective of a service provider, as shown in Fig.~\ref{fig:gsaas_system_model}. We consider a GSaaS service provider operating a set of geo-distributed ground stations and serving a set of individual satellites that request downlink services. Each satellite submits downlink tasks to the GSaaS provider before the corresponding deadlines, and the provider actively determines task admission, ground-station assignment, task scheduling, bandwidth allocation, and payments.

\begin{figure}[t]
    \centering
    \includegraphics[width=1\linewidth]{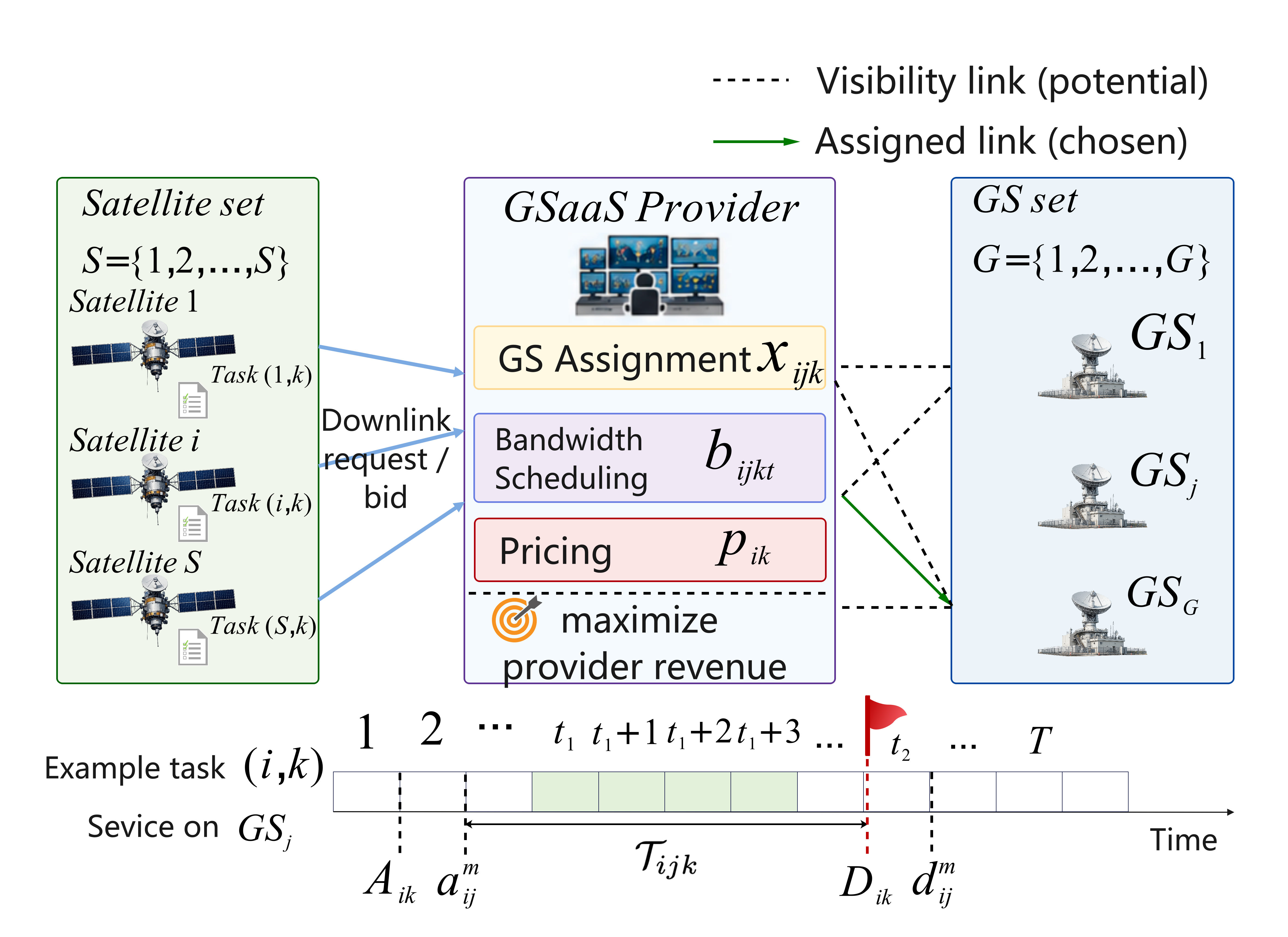}
    \caption{
System model of the GSaaS resource allocation framework. 
Satellites submit downlink requests and bids to the GSaaS provider, which determines task admission, GS assignment \(x_{ijk}\), bandwidth allocation \(b_{ijkt}\), and pricing \(p_{ik}\). Dashed links indicate potential visibility between satellites and ground stations, while solid links indicate chosen assignments. The bottom timeline shows how the feasible service window \(T_{ijk}\) of task \((i,k)\) is formed by intersecting its lifetime \([A_{ik},D_{ik}]\) with the visibility window \([a^m_{ij},d^m_{ij}]\).
}
    \label{fig:gsaas_system_model}
\end{figure}

\subsection{System Entities and Time Model}

Let $\mathcal{S} = \{1, 2, \dots, S\}$ denote the set of satellites that generate downlink requests within a scheduling horizon. Similarly, let $\mathcal{G} = \{1, 2, \dots, G\}$ denote the set of commercial Ground Stations (GSs) operated by the service provider.

Following existing work~\cite{6094268}, we assume that time is discretized into a set of equal-duration slots, indexed by $\mathcal{T} = \{1, 2, \dots, T\}$, where each slot has duration $\Delta t$. This common discretization maps heterogeneous visibility windows onto a finite scheduling horizon and can be extended to irregular time slots by using slot-dependent durations.

\subsection{Satellite Downlink Requests}

Each downlink request is modeled as a persistent communication task characterized by its arrival time, deadline, and data demand:
\begin{equation}
(A_{ik}, D_{ik}, \tau_{ik}), \quad i \in \mathcal{S}, k \in \mathcal{K}_i .
\end{equation}
Here, $A_{ik}$ denotes the arrival time slot of the $k$-th downlink task generated by satellite $i$, $D_{ik}$ denotes the deadline by which the task must be completed, and $\tau_{ik}$ denotes the amount of data to be downlinked.
We assume that $\tau_{ik}$ is declared by the task owner when the request is submitted. This assumption is suitable for buffered downlink tasks, such as stored imagery, telemetry, and mission files, whose data volume is known before scheduling.

Each downlink task can be requested by a different application-level user through satellite tasking platforms such as Planet or SkyFi, and thus may have its own private valuation~\cite{planet_tasking_api,skyfi_tasking}. We therefore treat each task owner as an independent auction participant and define value, utility, and payment at the task level; a task is paid only if its cumulative downlink reaches $\tau_{ik}$ by deadline $D_{ik}$.

\subsection{Commercial Ground Stations and Visibility}

We consider a multi-GS setting in which a GSaaS provider operates a set of geographically distributed ground stations $\mathcal{G}$. Each ground station $j \in \mathcal{G}$ is equipped with available downlink bandwidth $B_j$. A ground station can serve multiple satellites in the same time slot through Frequency-Division Multiple Access (FDMA)~\cite{myung2008single}. We consider a finite planning horizon consisting of the time slots in $\mathcal{T}$. Due to orbital dynamics and line-of-sight constraints, a satellite can communicate with a ground station only during specific contact intervals within this planning horizon. Since a satellite may have multiple contacts with the same ground station during the horizon, we denote the set of visibility windows $\mathcal{C}_{ij}$ between satellite $i$ and ground station $j$ as
\begin{equation}
\mathcal{C}_{ij} = \{[a_{ij}^{m}, d_{ij}^{m}]\}_{m=1}^{M_{ij}},
\end{equation}
where $M_{ij}$ is the number of visibility windows between satellite $i$ and ground station $j$, and $[a_{ij}^{m}, d_{ij}^{m}]$ denotes the $m$-th visibility window.
These visibility windows are deterministically derived from the predicted satellite orbit and the geographic location of the ground station using line-of-sight geometry.

For each downlink task $(i,k)$ and ground station $j$, we define its feasible service window as the set of time slots that are both within the task lifetime and covered by at least one visibility window:
\begin{equation}\label{def:fea-window}
\mathcal{T}_{ijk}
\triangleq
\left( \bigcup_{m=1}^{M_{ij}} [a_{ij}^{m}, d_{ij}^{m}] \right)
\cap [A_{ik}, D_{ik}].
\end{equation}

\subsection{Bandwidth Allocation and Ground Station Assignment}

We consider two scheduling decisions made by the GSaaS provider: ground station assignment and bandwidth allocation. Let \(b_{ijkt} \ge 0\) denote the bandwidth allocated by GS \(j\) to downlink task \((i,k)\) at time slot \(t\). The ground station assignment variable is defined as
\begin{equation}
\label{assignment_variable}
x_{ijk}\in\{0,1\},
\quad \forall i\in\mathcal{S},\; k\in\mathcal{K}_i,\; j\in\mathcal{G},
\end{equation}
where \(x_{ijk}=1\) indicates that downlink task \((i,k)\) is assigned to GS \(j\), otherwise \(x_{ijk}=0\). The bandwidth allocation and ground station assignment decisions are subject to the following physical and scheduling constraints.

At each time slot, the aggregate bandwidth allocated by each GS cannot exceed its available bandwidth:
\begin{equation}
\label{total_bandwidth_constraint}
\sum_{i \in \mathcal{S}} \sum_{k \in \mathcal{K}_i} b_{ijkt} \le B_j,
\quad \forall j \in \mathcal{G},\; t \in \mathcal{T}.
\end{equation}

A task can receive bandwidth from a GS only if it is assigned to that GS. This coupling between bandwidth allocation and GS assignment is enforced by
\begin{equation}
\label{bandwidth_constraint}
0 \le b_{ijkt} \le x_{ijk} B_j,
\quad \forall i \in \mathcal{S},\; k \in \mathcal{K}_i,\; j \in \mathcal{G},\; t \in \mathcal{T}.
\end{equation}

In addition, a GS can serve a task only when the satellite is visible to the GS and the task has arrived but not yet expired. Therefore, bandwidth allocation is restricted to the feasible service window \(\mathcal{T}_{ijk}\):
\begin{equation}
\label{visibility_constraint}
b_{ijkt} = 0,
\quad
\forall t \notin \mathcal{T}_{ijk},
\quad
\forall i \in \mathcal{S},\; k \in \mathcal{K}_i,\; j \in \mathcal{G}.
\end{equation}

Following practical GSaaS contact-reservation operations~\cite{aws_contact}, we model each downlink task as being explicitly reserved at a specific ground station. This does not restrict a satellite to a single ground station: different tasks from the same satellite may be assigned to different ground stations across different contact opportunities. Therefore, each downlink task can be assigned to at most one primary GS, which is captured by
\begin{equation}
\label{single_gs_constraint}
\sum_{j \in \mathcal{G}} x_{ijk} \le 1,
\quad \forall i \in \mathcal{S},\; k \in \mathcal{K}_i.
\end{equation}

\subsection{Request Completion}

Let $\eta_{ij}$ denote the spectral efficiency (bps/Hz) of the downlink between satellite $i$ and ground station $j$.
We model $\eta_{ij}$ as a fixed effective spectral efficiency for each satellite--GS pair, representing an average link rate over the scheduling horizon.
The amount of data transmitted for downlink task $(i,k)$ from satellite $i$ to GS $j$ during time slot $t$ is given by
\begin{equation}
\label{data_transmit}
\Delta \tau_{ijkt} = b_{ijkt} \cdot \eta_{ij} \cdot \Delta t,
\forall i \in \mathcal{S},\; k \in \mathcal{K}_i,\; j \in \mathcal{G},\; t \in \mathcal{T}.
\end{equation}

Since a downlink task may span multiple time slots, we track its remaining data volume over time. Let $\tau_{ik}^{\mathrm{rem}}(t)$ denote the remaining amount of data of task $(i,k)$ at the beginning of time slot $t$. Its evolution is given by
\begin{equation}
\label{remaining_data}
\begin{aligned}
\tau_{ik}^{\mathrm{rem}}(t+1)
=
\max\Bigg\{
0,\;
\tau_{ik}^{\mathrm{rem}}(t)
-
\sum_{j \in \mathcal{G}} \Delta \tau_{ijkt}
\Bigg\},
\\
\forall i \in \mathcal{S},\; k \in \mathcal{K}_i .
\end{aligned}
\end{equation}

The completion indicator is defined according to whether the remaining data volume reaches zero before the deadline:
\begin{equation}
\label{completed_constrain}
y_{ik}
=
\begin{cases}
1, & \text{if } \tau_{ik}^{\mathrm{rem}}(D_{ik}+1)=0,\\
0, & \text{otherwise.}
\end{cases}
\end{equation}

If a downlink task cannot be completed by its deadline, it will be dropped and yield no revenue to the service provider.

\subsection{Pricing, Utilities, and Economic Constraints}

We model GSaaS pricing as an auction, where the GSaaS provider acts as the seller/auctioneer and downlink task owners act as buyers competing for limited downlink resources. Each task owner submits one downlink request and a unit bid $\mathrm{bid}_{ik}$ for completing this request. Following the standard single-parameter private-value setting~\cite{a2fd2512-18e4-3e86-a4f2-2f4443994eb2}, each task $(i,k)$ has a private unit value $v_{ik}$, which represents the task owner's valuation for successfully completing the task. This value is private information of the task owner and is not directly observable by the GSaaS provider. As a result, a task owner may strategically misreport its willingness to pay through $\mathrm{bid}_{ik}$, either by underbidding to reduce its potential payment or by overbidding to increase its chance of being admitted and scheduled.

For each admitted task $(i,k)$, the GSaaS provider charges a unit service price $p_{ik}$. Since the task requires the transmission of $\tau_{ik}$ units of data, the total payment charged to the owner of task $(i,k)$ is $P_{ik}=p_{ik}\tau_{ik}y_{ik}$. If the task is not completed before its deadline, i.e., $y_{ik}=0$, no payment is collected.

The utility of the owner of task $(i,k)$ is defined as
\begin{equation}
\label{user_utility}
u^{task}_{ik}=
\begin{cases}
(v_{ik}-p_{ik})\tau_{ik}, & \text{if } y_{ik}=1,\\
0, & \text{otherwise}.
\end{cases}
\end{equation}

The GSaaS provider uses a reserve unit price to capture its minimum acceptable payment for serving downlink traffic. In this paper, we use a uniform reserve price $\rho$ for all tasks as a simplifying model, representing the provider's operating expense or opportunity cost per unit downlink data. This model can be extended to task-dependent or GS-dependent reserve prices when heterogeneous operating costs are considered. Under the proposed mechanism, a task is eligible for admission only if its bid satisfies $\mathrm{bid}_{ik}\ge \rho$.

The utility of the GSaaS provider obtained from serving task $(i,k)$ is defined as
\begin{equation}
u^{GS}_{ik}
=
\begin{cases}
(p_{ik}-\rho)\tau_{ik}, & \text{if } y_{ik}=1,\\
0, & \text{otherwise}.
\end{cases}
\end{equation}
Thus, for any admitted and completed task, the proposed mechanism enforces the transaction price to satisfy
\begin{equation}
\label{price_constraint}
\rho \le p_{ik} \le \mathrm{bid}_{ik}, \quad \forall (i,k)\text{ with }y_{ik}=1.
\end{equation}
This constraint ensures that each completed winning task is charged at least the provider's reserve price and no more than its reported bid.

We evaluate the mechanism by individual rationality, truthfulness, and reserve-price protection. Individual rationality requires truthful task owners to obtain non-negative utility; truthfulness, or dominant-strategy incentive compatibility, requires truthful bidding to maximize each task owner's utility regardless of others' bids; and reserve-price protection requires every completed winning task to be charged at least $\rho$.

\subsection{Optimization Objective}

Based on the pricing and utility model defined above, the GSaaS provider aims to maximize its total revenue. The revenue maximization problem can be formulated as
\begin{equation}
\label{eq:opt_problem}
\begin{aligned}
\max_{\{b_{ijkt},\,x_{ijk},\,p_{ik}\}} \quad
 \sum_{i\in \mathcal{S}}&
\sum_{k\in \mathcal{K}_i}
p_{ik}\tau_{ik}y_{ik} \\
\text{s.t.}\quad
\text{Constraints } 
&\eqref{assignment_variable}-
\eqref{completed_constrain},
\eqref{price_constraint}.
\end{aligned}
\end{equation}

In Problem~\eqref{eq:opt_problem}, the GSaaS provider jointly optimizes bandwidth allocation, ground-station assignment, and payment determination to maximize total revenue, subject to assignment, bandwidth-capacity, visibility-window, task-completion, and pricing constraints.

\begin{theorem}
Problem~\eqref{eq:opt_problem} is NP-hard.
\end{theorem}

\begin{proof}
We reduce the 0--1 knapsack problem to Problem~\eqref{eq:opt_problem}. Consider an arbitrary knapsack instance with item set $\mathcal N=\{1,\ldots,N\}$, item weights $\{w_n\}$, item values $\{q_n\}$, and capacity $C$. We construct a GSaaS instance with one ground station and one time slot. The ground-station bandwidth is set to $B_1=C$, and we set $\Delta t=1$ and $\eta_{n1}=1$ for all satellites $n$.

For each item $n\in\mathcal N$, we create one satellite $i=n$ with a single downlink task $k=1$, i.e., $\mathcal{S}=\{1,\ldots,N\}$ and $\mathcal{K}_n=\{1\}$. The task $(n,1)$ has demand $\tau_{n1}=w_n$, arrival time $A_{n1}=1$, deadline $D_{n1}=1$, feasible service window $\mathcal{T}_{n,1,1}=\{1\}$, unit bid $\mathrm{bid}_{n1}=q_n/w_n$, and reserve price $\rho=0$.

In this constructed instance, completing task $(n,1)$ requires exactly $w_n$ units of bandwidth in the only slot, while the total available bandwidth is $C$. Thus, any feasible set of completed tasks corresponds exactly to a feasible knapsack selection. Moreover, by setting $p_{n1}=\mathrm{bid}_{n1}$, the revenue obtained from task $(n,1)$ is
\begin{equation}
    p_{n1}\tau_{n1}
=
\mathrm{bid}_{n1}\tau_{n1}
=
\frac{q_n}{w_n}w_n
=
q_n.
\end{equation}

Therefore, maximizing the GSaaS revenue is equivalent to maximizing the total value of selected knapsack items. Since 0--1 knapsack is NP-hard, Problem~\eqref{eq:opt_problem} is NP-hard.
\end{proof}

\section{STAR-GS Design}
\label{sec:alg_design}

Beyond its NP-hardness, Problem~\eqref{eq:opt_problem} tightly couples task admission, ground-station assignment, bandwidth allocation, and payment determination under visibility and deadline constraints. STAR-GS addresses this coupling through bid-aware admission, best-fit ground-station assignment, EDF-based schedulability testing, and critical-payment pricing. Intuitively, bid-aware admission prioritizes valuable tasks, EDF prevents infeasible admissions, best-fit assignment preserves flexible GS resources, and critical payments ensure truthful bidding. For feasibility checking, STAR-GS converts fractional bandwidth allocation into equivalent preemptive full-bandwidth time-sharing under a fixed GS assignment, thereby reducing bandwidth feasibility to deadline schedulability.

\subsection{From Bandwidth Allocation to Deadline Scheduling}
\label{subsec:bandwidth_to_scheduling}

The optimization problem in~\eqref{eq:opt_problem} is formulated through the bandwidth allocation variables \(b_{ijkt}\). For a fixed GS assignment, however, downlink feasibility can be viewed as a real-time scheduling problem: a GS acts as a processor, its bandwidth as processor capacity, a downlink task as a real-time job, the task demand as the processing requirement, and the feasible service window \(\mathcal{T}_{ijk}\) as the job availability interval.

Consider a fixed task \((i,k)\) assigned to GS \(j\). If GS \(j\) allocates its full bandwidth \(B_j\) to this task for one slot, the transmitted data is \(B_j \eta_{ij} \Delta t\). Thus, the normalized processing requirement of task \((i,k)\) on GS \(j\), measured in full-bandwidth slots, is
\begin{equation}
\label{eq:processing_requirement}
c_{ijk}
=
\frac{\tau_{ik}}{B_j \eta_{ij}\Delta t}.
\end{equation}
The task is feasible on GS \(j\) only if it receives at least \(c_{ijk}\) full-bandwidth slot equivalents within its feasible service window \(\mathcal{T}_{ijk}\) in \eqref{def:fea-window}.

This scheduling view avoids explicitly searching over all fractional bandwidth allocations: if GS $j$ allocates bandwidth $b_{ijkt}$ to task $(i,k)$ in slot $t$, this is equivalent to serving the task with full bandwidth $B_j$ for an $\alpha_{ijkt}=b_{ijkt}/B_j$ fraction of the slot. By the per-slot bandwidth capacity constraint in \eqref{total_bandwidth_constraint}, we have
\begin{equation}
\sum_{i\in\mathcal{S}}\sum_{k\in\mathcal{K}_i} \alpha_{ijkt}\leq 1 .
\end{equation}
Therefore, any feasible fractional allocation within a slot can be serialized into a full-bandwidth time-sharing schedule without changing the transmitted data. Conversely, any full-bandwidth time-sharing schedule induces a feasible bandwidth allocation by setting the average bandwidth in each slot. Hence, bandwidth allocation and preemptive full-bandwidth scheduling are equivalent for feasibility analysis.

\begin{lemma}[Bandwidth--Scheduling Equivalence]
\label{lem:bandwidth_scheduling_equivalence}
For a fixed GS $j$ and a fixed assigned task set $\mathcal{W}_j$, there exists a feasible bandwidth allocation that completes every task in $\mathcal{W}_j$ before its deadline if and only if there exists a feasible preemptive full-bandwidth time-sharing schedule that completes the same task set before the same deadlines.
\end{lemma}

\begin{proof}
We prove both directions.

First, suppose there exists a feasible bandwidth allocation $\{b_{ijkt}\}$ for GS $j$. In each slot $t$, define $\alpha_{ijkt}=b_{ijkt}/B_j$. By the bandwidth-capacity constraint, $\sum_{(i,k)\in\mathcal{W}_j}\alpha_{ijkt}\le 1$. We can therefore replace the fractional bandwidth allocation in slot $t$ by a full-bandwidth service interval of length $\alpha_{ijkt}\Delta t$ for each task $(i,k)$. This transformation preserves the transmitted data, since
\begin{equation}
    B_j\eta_{ij}\alpha_{ijkt}\Delta t=b_{ijkt}\eta_{ij}\Delta t.
\end{equation}
It is performed within the same slot, so visibility and deadline constraints are preserved.

Conversely, suppose there exists a feasible preemptive full-bandwidth time-sharing schedule. If task $(i,k)$ is served by GS $j$ for a fraction $\alpha_{ijkt}$ of slot $t$, we construct the corresponding bandwidth allocation as $b_{ijkt}=\alpha_{ijkt}B_j$. Since the time-sharing schedule uses at most one full-bandwidth server in each slot, $\sum_{(i,k)}\alpha_{ijkt}\le 1$ and thus $\sum_{(i,k)}b_{ijkt}\le B_j$. The transmitted data in slot $t$ is $b_{ijkt}\eta_{ij}\Delta t$, matching the service provided by the schedule. Therefore, the constructed allocation satisfies the capacity, visibility, and completion constraints.

Therefore, feasible bandwidth allocation and feasible preemptive full-bandwidth scheduling are equivalent for a fixed GS assignment.
\end{proof}

Lemma~\ref{lem:bandwidth_scheduling_equivalence} transforms the bandwidth allocation problem into a preemptive single-processor scheduling problem with task-specific service windows. For a fixed GS assignment, each GS is viewed as a processor, and each downlink task is characterized by its processing requirement, deadline, and feasible service windows. We adopt EDF as a deadline-aware scheduling rule that prioritizes currently serviceable tasks with earlier deadlines. Although EDF is not generally optimal under arbitrary visibility constraints, it provides an efficient and deterministic feasibility-checking procedure, making it suitable for online GSaaS admission decisions.

\subsection{EDF-Based Schedulability Test}
\label{subsec:edf_check}

Algorithm~\ref{alg:edf_check} presents \textsc{EDF\_Check}, an EDF-based schedulability sufficient schedulability test. Given a GS \(j\) and a candidate task set \(W_j\), \textsc{EDF\_Check} determines whether all tasks in \(W_j\) can be completed by their deadlines using GS \(j\). It also returns the induced bandwidth allocation and the residual slack of the GS.

At each slot \(t\), \textsc{EDF\_Check} first identifies the set of currently serviceable unfinished tasks:
\begin{equation}
\label{eq:edf_active_set}
\mathcal{A}_j(t)
=
\left\{
(i,k)\in W_j:
t\in \mathcal{T}_{ijk},
\tau^{\mathrm{rem}}_{ik}(t)>0
\right\}.
\end{equation}

If \(\mathcal{A}_j(t)\neq\emptyset\), EDF selects the task with the earliest deadline:
\begin{equation}
\label{eq:edf_select}
(i^\star,k^\star)
=
\arg\min_{(i,k)\in \mathcal{A}_j(t)} D_{ik}.
\end{equation}

Ties are resolved by a fixed task-ID rule, which is independent of bids. GS \(j\) then serves the selected task using full bandwidth until the task is completed, the slot is exhausted, or the task becomes no longer serviceable. If the selected task completes before the end of the slot, the remaining capacity in that slot can be used to serve the next earliest-deadline task. This implements the full-bandwidth time-sharing interpretation.

The EDF schedule is accepted if every task completes before its deadline:
\begin{equation}
\tau^{\mathrm{rem}}_{ik}(D_{ik}+1)=0,\quad \forall (i,k)\in W_j.
\end{equation}
Otherwise, \textsc{EDF\_Check} returns false. The residual slack of GS \(j\) is defined as the total unused full-bandwidth slot equivalent:
\begin{equation}
\label{eq:slack_new}
\mathrm{Slack}_j
=
\sum_{t\in\mathcal{T}} I_j(t),
\end{equation}
where \(I_j(t)\) denotes the unused fraction of slot \(t\) after EDF scheduling. A smaller residual slack means that the task set fits more tightly on that GS, leaving other GSs less fragmented and more available for future tasks.

\begin{algorithm}[t]
\caption{\mbox{\textsc{EDF\_Check}: EDF-Based Schedulability Test}}
\label{alg:edf_check}
\KwIn{
Assigned task set \(W_j\) on GS \(j\);
feasible service windows
\(\{\mathcal{T}_{ijk}\}_{(i,k)\in W_j}\);
GS bandwidth \(B_j\)
}
\KwOut{
Feasibility indicator \(\mathrm{Flag}\);
bandwidth allocation \(\{b_{ijkt}\}\);
residual slack \(\mathrm{Slack}_j\)
}

Initialize
\(\tau^{\mathrm{rem}}_{ik}\leftarrow\tau_{ik}\),
\(b_{ijkt}\leftarrow0\), and \(I_j(t)\leftarrow0\),
\(\forall(i,k)\in W_j,\ t\in\mathcal{T}\)\;

\For{each slot \(t\in\mathcal{T}\)}{
    \(C\leftarrow1\)
    \tcp*[r]{Remaining full-bandwidth fraction}

    \While{\(C>0\)}{
        Form \(\mathcal{A}_j(t)\) according to
        Eq.~\eqref{eq:edf_active_set}
        \;

        \If{\(\mathcal{A}_j(t)=\emptyset\)}{
            \(I_j(t)\leftarrow C\)\;
            \textbf{break}\;
        }

        Select \((i^\star,k^\star)\) according to
        Eq.~\eqref{eq:edf_select},
        with ties resolved by task IDs
        \;

        Serve \((i^\star,k^\star)\) until it completes
        or \(C\) is exhausted; update \(b\),
        \(\tau^{\mathrm{rem}}\), and \(C\) according to
        Eqs.~\eqref{data_transmit} and~\eqref{remaining_data}
        \;
    }
}

Compute \(\mathrm{Slack}_j\) according to
Eq.~\eqref{eq:slack_new}
\;

\If{
    Eq.~\eqref{completed_constrain} is not satisfied
    for some \((i,k)\in W_j\)
}{
    \Return false, \(\emptyset\), \(\emptyset\)\;
}

\Return true, \(b_{ijkt}\), \(\mathrm{Slack}_j\)\;

\end{algorithm}

\begin{lemma}[Sufficient Feasibility of EDF-Based Schedulability Test]
\label{lem:edf_check_sufficient}
For a fixed GS \(j\) and a fixed assigned task set \(W_j\), if
\textsc{EDF\_Check}\((W_j,j)\) returns true, then there exists a feasible
preemptive full-bandwidth schedule that completes all tasks in \(W_j\)
before their deadlines.
\end{lemma}

\begin{proof}
For a fixed GS \(j\), \textsc{EDF\_Check}\((W_j,j)\) explicitly constructs
a preemptive full-bandwidth schedule.

At each time slot, the algorithm only allocates service to tasks that
belong to their feasible service windows \(T_{ijk}\). Moreover, the
allocated full-bandwidth fractions satisfy the GS capacity constraint,
and the algorithm terminates successfully only when every task satisfies
\(\tau^{\mathrm{rem}}_{ik}(D_{ik}+1)=0\).

Therefore, the schedule generated by \textsc{EDF\_Check}\((W_j,j)\)
satisfies the visibility constraints, bandwidth constraints, and deadline
constraints. Hence, a feasible schedule exists.
\end{proof}

\subsection{Feasibility-Aware Primary-GS Assignment}
\label{subsec:feasible_assignment}

Algorithm~\ref{alg:feasible_assign_new} presents \textsc{Feasible\_Assign}, which constructs a feasibility-aware primary-GS assignment. The EDF-based test above determines whether a given set of tasks can be scheduled on a fixed GS. We next describe how STAR-GS assigns tasks to GSs. Following the primary-GS constraint in~\eqref{single_gs_constraint}, each admitted task is assigned to at most one GS. This design matches primary contact reservation and avoids cross-GS state synchronization and coordination overhead.\footnote{Multi-GS task splitting may further improve schedulability by combining disjoint contact opportunities, but supporting it would require corresponding extensions to the assignment, scheduling, and payment mechanisms.}

Given a candidate task set \(W\), \textsc{Feasible\_Assign} attempts to construct a feasible primary-GS assignment and the corresponding EDF schedules. It processes tasks in non-increasing order of reported bids:
\begin{equation}
\pi(W)=\big((i_1,k_1),(i_2,k_2),\ldots,(i_{|W|},k_{|W|})\big),
\end{equation}
where
\begin{equation}
\label{eq:bid_order}
\mathrm{bid}_{i_1k_1}\geq
\mathrm{bid}_{i_2k_2}\geq
\cdots\geq
\mathrm{bid}_{i_{|W|}k_{|W|}} .
\end{equation}
The use of bid order is revenue-aware: tasks with higher reported values are considered earlier, so scarce GS resources are first reserved for tasks that contribute more to the provider's potential revenue. Other fixed priority rules, such as earliest-deadline-first admission or throughput-first admission, can also be used as heuristics, but they do not directly align the admission order with the revenue objective. Importantly, ties are resolved by a fixed bid-independent rule, which is necessary for deterministic monotonicity in the auction mechanism.

Let \(W_j\) denote the set of tasks already assigned to GS \(j\). When considering task \((i,k)\), STAR-GS tentatively adds it to each candidate GS \(j\) and checks whether \(W_j\cup\{(i,k)\}\) remains schedulable using \textsc{EDF\_Check}. The acceptable candidate set is defined as
\begin{equation}
\label{eq:feasible_gs_set}
\begin{aligned}
\mathcal{F}_{ik}
=
\big\{
j\in&\mathcal{G}:\;
\mathcal{T}_{ijk}\neq\emptyset, \\
&
\textsc{EDF\_Check}(W_j\cup\{(i,k)\},j)=\mathrm{true}
\big\}.
\end{aligned}
\end{equation}
If \(\mathcal{F}_{ik}=\emptyset\), the candidate task set \(W\) is declared infeasible. Otherwise, STAR-GS selects a GS according to a best-fit rule:
\begin{equation}
\label{eq:best_fit_gs}
g(i,k)
=
\arg\min_{j\in\mathcal{F}_{ik}}
\widehat{\mathrm{Slack}}_j,
\end{equation}
where \(\widehat{\mathrm{Slack}}_j\) is the residual slack returned by \textsc{EDF\_Check} after tentatively adding \((i,k)\) to GS \(j\). This best-fit rule is inspired by classical bin-packing heuristics~\cite{doi:10.1137/0203025}: the task is placed on the feasible GS where it fits most tightly, so that less-loaded GSs remain available for later tasks with more restrictive visibility windows.

\begin{algorithm}[t]
\caption{\textsc{Feasible\_Assign}: Feasibility-Aware Primary-GS Assignment}
\label{alg:feasible_assign_new}

\KwIn{
Candidate task set \(W\);
GS set \(\mathcal{G}\);
feasible service windows \(\{\mathcal{T}_{ijk}\}\);
task bids \(\{\mathrm{bid}_{ik}\}\)
}
\KwOut{
Feasibility indicator \(\mathrm{flag}\);
GS assignment \(\{x_{ijk}\}\);
bandwidth allocation \(\{b_{ijkt}\}\);
residual slack
\(\{\mathrm{Slack}_j\}_{j\in\mathcal{G}}\)
}

Initialize
\(W_j\leftarrow\emptyset,\ \forall j\in\mathcal{G}\),
and \(x_{ijk}\leftarrow0,\ b_{ijkt}\leftarrow0\)\;

Sort \(W\) according to Eq.~\eqref{eq:bid_order}
to obtain \(\pi(W)\)
\;

\For{each \((i,k)\in\pi(W)\)}{
    Construct \(\mathcal{F}_{ik}\) according to
    Eq.~\eqref{eq:feasible_gs_set},
    retaining the tentative
    \(\widehat{\mathrm{Slack}}_j\)
    returned by \textsc{EDF\_Check}
    \;

    \If{\(\mathcal{F}_{ik}=\emptyset\)}{
        \Return false, \(\emptyset\),
        \(\emptyset\), \(\emptyset\)\;
    }

    Select \(g(i,k)\) according to
    Eq.~\eqref{eq:best_fit_gs},
    with ties resolved by the smallest GS ID
    \;

    \(W_{g(i,k)}
    \leftarrow W_{g(i,k)}\cup\{(i,k)\}\)\;

    \(x_{i\,g(i,k)\,k}\leftarrow1\)\;
}

\For{each \(j\in\mathcal{G}\)}{
    \((\mathrm{flag},b_{ijkt},\mathrm{Slack}_j)
    \leftarrow\textsc{EDF\_Check}(W_j,j)\)\;

    \If{\(\mathrm{flag}=\mathrm{false}\)}{
        \Return false, \(\emptyset\),
        \(\emptyset\), \(\emptyset\)\;
    }
}

\Return true,
\(x_{ijk}\),
\(b_{ijkt}\),
\(\{\mathrm{Slack}_j\}_{j\in\mathcal{G}}\)\;

\end{algorithm}

\subsection{Task Admission and Critical Payment Pricing}
\label{subsec:truthful_auction}

Algorithm~\ref{alg:truthful_admission_new} presents the task admission and critical-payment pricing procedure of STAR-GS. Each task \((i,k)\) has a private unit value \(v_{ik}\) and submits a unit bid \(\mathrm{bid}_{ik}\). The GSaaS provider also has a reserve unit price \(\rho\), which represents the minimum acceptable price for serving the task. STAR-GS first removes all tasks whose bids are below their reserve prices:
\begin{equation}
\label{eq:reserve_eligible_set}
\mathcal{R}^{+}
=
\{(i,k)\in\mathcal{R}:\mathrm{bid}_{ik}\geq \rho\}.
\end{equation}

The remaining tasks are processed in non-increasing bid order. For each task, STAR-GS tentatively inserts it into the current winner set and calls \textsc{Feasible\_Assign}. The task is admitted only if the enlarged winner set remains feasible.

After the winner set is determined, STAR-GS computes a critical payment for each winning task. The critical payment is the smallest bid at which the task would still win when all other bids are fixed. In implementation, STAR-GS constructs an ordered candidate bid list within this interval, consisting of the reserve price \(\rho\), the winner's submitted bid \(\mathrm{bid}_{ik}\), and other submitted bids in \([\rho,\mathrm{bid}_{ik}]\), and then performs binary search over this finite list. This threshold-payment rule is essential for \textbf{truthfulness} in single-parameter mechanisms.

\begin{algorithm}[t]
\caption{Truthful Admission and Critical Payment Pricing}
\label{alg:truthful_admission_new}

\KwIn{
Task set \(\mathcal{R}\);
task bids \(\{\mathrm{bid}_{ik}\}\);
reserve price \(\rho\)
}
\KwOut{
Winner set \(\mathcal{W}\);
critical payments \(\{p_{ik}\}\);
GS assignment \(x\);
bandwidth allocation \(b\)
}

Initialize
\(\mathcal{W}\leftarrow\emptyset\) and
\(p_{ik}\leftarrow0,\ 
\forall(i,k)\in\mathcal{R}\)\;

\tcp{Winner-selection phase}

Form \(\mathcal{R}^{+}\) according to
Eq.~\eqref{eq:reserve_eligible_set}
\;

Sort \(\mathcal{R}^{+}\) according to
Eq.~\eqref{eq:bid_order}
to obtain \(\pi(\mathcal{R}^{+})\)
\;

\For{each \((i,k)\in\pi(\mathcal{R}^{+})\)}{
    \(\mathcal{W}^{\mathrm{tmp}}
    \leftarrow
    \mathcal{W}\cup\{(i,k)\}\)\;

    \((\mathrm{flag},x,b,\mathrm{Slack})
    \leftarrow
    \textsc{Feasible\_Assign}
    (\mathcal{W}^{\mathrm{tmp}})\)\;

    \If{\(\mathrm{flag}=\mathrm{true}\)}{
        \(\mathcal{W}
        \leftarrow
        \mathcal{W}^{\mathrm{tmp}}\)\;
    }
}

\((\mathrm{flag},x,b,\mathrm{Slack})
\leftarrow
\textsc{Feasible\_Assign}(\mathcal{W})\)\;

\tcp{Critical-payment phase}

\For{each \((i,k)\in\mathcal{W}\)}{
    Build the candidate bid list within
    \([\rho,\mathrm{bid}_{ik}]\):
    
    \(\mathcal{B}_{ik}\leftarrow
    \{\rho\}\cup\{\mathrm{bid}_{ik}\}
    \cup
    \big\{
    \mathrm{bid}_{i'k'}:
    (i',k')\in\mathcal{R},
    (i',k')\neq(i,k),
    \rho\leq\mathrm{bid}_{i'k'}
    \leq\mathrm{bid}_{ik}
    \big\}\)\;

    Sort \(\mathcal{B}_{ik}\) in non-decreasing order as
    \((\beta_1,\beta_2,\ldots,
    \beta_{|\mathcal{B}_{ik}|})\)\;

    \(\ell\leftarrow1,\quad
    r\leftarrow|\mathcal{B}_{ik}|\)\;

    \While{\(\ell<r\)}{
        \(m\leftarrow
        \lfloor(\ell+r)/2\rfloor\)\;

        Replace \(\mathrm{bid}_{ik}\) by \(\beta_m\),
        keeping all other bids unchanged\;

        Rerun the winner-selection phase to obtain
        \(\mathcal{W}^{(\beta_m)}\)\;

        \If{\((i,k)\in\mathcal{W}^{(\beta_m)}\)}{
            \(r\leftarrow m\)\;
        }
        \Else{
            \(\ell\leftarrow m+1\)\;
        }
    }

    \(p_{ik}\leftarrow\beta_{\ell}\)\;
}

\Return
\(\mathcal{W}\),
\(\{p_{ik}\}_{(i,k)\in\mathcal{W}}\),
\(x\), \(b\)\;

\end{algorithm}

\subsection{Theoretical Properties}
\label{subsec:theoretical_properties}

We analyze STAR-GS under the single-parameter private-value model in auction theory~\cite{myerson1981optimal}, where each task \((i,k)\) has a private value \(v_{ik}\) and submits a bid \(\mathrm{bid}_{ik}\), while task demand, deadlines, visibility windows, GS bandwidths, and channel rates are public or verifiable. \textbf{Truthfulness} means that reporting \(\mathrm{bid}_{ik}=v_{ik}\) is a dominant strategy for every task.

\begin{lemma}[Deletion Monotonicity of \textsc{EDF\_Check}]
\label{lem:edf_deletion_new}
For a fixed GS \(j\), if \(\textsc{EDF\_Check}(V,j)=\mathrm{true}\), then for any subset \(U\subseteq V\), \(\textsc{EDF\_Check}(U,j)=\mathrm{true}\).
\end{lemma}

\begin{proof}
Removing tasks does not reduce the bandwidth of GS \(j\), shorten the feasible service window of any remaining task, or increase its demand. Hence, any subset \(U\subseteq V\) is no harder to schedule than \(V\). Since \(V\) is EDF-feasible, \textsc{EDF\_Check} also completes all tasks in \(U\) by their deadlines, and therefore
\(\textsc{EDF\_Check}(U,j)=\mathrm{true}\).
\end{proof}

\begin{lemma}[Monotonicity of Winner Selection]
\label{lem:winner_monotone_new}
The winner-selection rule of STAR-GS is monotone in each task's bid. Fixing all other bids, for any task $e=(i,k)$, the following two properties hold: (i) if $e$ wins with bid $\mathrm{bid}_{ik}$, then it also wins with any higher bid $\mathrm{bid}'_{ik}>\mathrm{bid}_{ik}$; (ii) if $e$ loses with bid $\mathrm{bid}_{ik}$, then it also loses with any lower bid $\mathrm{bid}'_{ik}<\mathrm{bid}_{ik}$.
\end{lemma}

\begin{proof}
Fix all bids except that of task $e=(i,k)$. We first prove the winner side. Suppose $e$ wins with bid $\mathrm{bid}_{ik}$. If its bid is increased to $\mathrm{bid}'_{ik}>\mathrm{bid}_{ik}$, then $e$ can only move earlier in the non-increasing bid order, while the relative order among all other tasks remains unchanged. When $e$ is considered under the higher bid, the set of tasks admitted before $e$ is no larger than the corresponding set under the lower bid. Since $e$ was feasible together with the previously admitted tasks under the lower bid, removing some previously admitted tasks cannot destroy feasibility by Lemma~\ref{lem:edf_deletion_new}. Hence, $e$ is still admitted under the higher bid.

We next prove the loser side. Suppose $e$ loses with bid $\mathrm{bid}_{ik}$, and consider any lower bid $\mathrm{bid}'_{ik}<\mathrm{bid}_{ik}$. If $e$ were to win under the lower bid, then by the winner-side monotonicity just proved, increasing its bid from $\mathrm{bid}'_{ik}$ back to $\mathrm{bid}_{ik}$ would imply that $e$ also wins with bid $\mathrm{bid}_{ik}$. This contradicts the assumption that $e$ loses with bid $\mathrm{bid}_{ik}$. Therefore, $e$ must also lose under any lower bid.
\end{proof}

\begin{lemma}[Critical Price as Threshold Payment]
\label{lem:critical_payment_new}
For any winning task $(i,k)$ produced by STAR-GS, the payment $p_{ik}$ is its critical price. That is, fixing all other bids, task $(i,k)$ wins if and only if $\mathrm{bid}_{ik}\ge p_{ik}$.
\end{lemma}

\begin{proof}
By Lemma~\ref{lem:winner_monotone_new}, the winner-selection rule is monotone in each task's bid. Therefore, by the characterization of single-parameter mechanisms~\cite{nisan2007algorithmic}, there exists a critical threshold price for each task. Since STAR-GS searches the ordered finite bid list by bisection and returns the smallest bid at which the task remains a winner, the returned payment $p_{ik}$ is exactly this critical price.
\end{proof}
\begin{theorem}
\label{thm:mechanism_properties_new}
Under the single-parameter private-value setting, STAR-GS is dominant-strategy incentive compatible, individually rational for task owners, satisfies the provider-side reserve-price constraint, and runs in polynomial time.
\end{theorem}

\begin{proof}
We first prove the provider-side reserve-price constraint. STAR-GS removes all tasks with $\mathrm{bid}_{ik}<\rho$ before admission. For every winning task, the critical payment satisfies $p_{ik}\ge \rho$. Therefore, every admitted and completed task is charged no less than the provider's reserve unit price:
\begin{equation}
p_{ik}\ge \rho,\quad \forall (i,k)\text{ with }y_{ik}=1.
\end{equation}
Equivalently, the provider-side utility for each admitted and completed task satisfies $u^{GS}_{ik}=(p_{ik}-\rho)\tau_{ik}\ge 0$. Thus, STAR-GS enforces provider-side reserve-price protection.

We next prove individual rationality. Under truthful bidding, $\mathrm{bid}_{ik}=v_{ik}$. For every winning task, the critical payment is no larger than its submitted bid, i.e., $p_{ik}\le \mathrm{bid}_{ik}=v_{ik}$. Therefore, $u^{task}_{ik}=(v_{ik}-p_{ik})\tau_{ik}\ge 0$. Losing tasks receive no service and make no payment, so their utility is zero. Hence, STAR-GS is individually rational for task owners.

We then prove dominant-strategy incentive compatibility. By Lemma~\ref{lem:winner_monotone_new}, the winner-selection rule is monotone in each task's bid. By Lemma~\ref{lem:critical_payment_new}, every winning task is charged its critical payment. By the characterization of single-parameter mechanisms~\cite{nisan2007algorithmic}, a monotone allocation rule combined with critical payments is dominant-strategy incentive compatible. Therefore, truthful bidding is a dominant strategy for every task owner.

Finally, STAR-GS runs in polynomial time because it performs polynomially many EDF-based schedulability checks, feasibility-aware assignments, and critical-payment searches over a finite bid set. Therefore, all stated properties hold.
\end{proof}

\section{Performance Evaluation}
\label{sec:simulation}
\subsection{Simulation Setup}

We evaluate STAR-GS using a discrete-time simulator driven by satellite--ground station access windows generated by Ansys Systems Tool Kit (STK), an industry-standard tool for high-fidelity LEO satellite mission simulation and access analysis~\cite{agi_stk}. We consider a Walker-style LEO constellation and a geographically distributed GSaaS network, with the default simulation parameters summarized in Table~\ref{tab:simulation-parameters}. Each experiment varies one parameter at a time and reports the average of $10$ runs with different random seeds.

\begin{table}[t]
\centering
\caption{Simulation Parameters}
\label{tab:simulation-parameters}
\footnotesize
\setlength{\tabcolsep}{1.5pt}
\renewcommand{\arraystretch}{0.88}
\begin{tabular}{ll}
\hline
Parameter & Value \\
\hline
Constellation type  \cite{mathworks_walker}  & Walker-style LEO constellation\\
Orbital planes \cite{mathworks_walker} & $3$ \\
Satellites per plane \cite{mathworks_walker} & $8$ \\
Total satellites \cite{mathworks_walker} & $24$ \\
Number of ground stations \cite{Duncan} & $6$ \\
Planning horizon \cite{363d0a634b234a96a2f672e30e105fe3} & $24$ hours \\
Slot duration \cite{363d0a634b234a96a2f672e30e105fe3} & $60$ s  \\
Number of slots \cite{363d0a634b234a96a2f672e30e105fe3} & $1440$ \\
Ground station bandwidth \cite{10901122} & $80$ Mbps \\
Task demand \cite{10858567} & $[5,20]$ GB \\
Unit bid \cite{4927011} & $[50,100]$ \\
Deadline slack \cite{10901122} & $[30,90]$ slots \\
Default number of tasks & $300$ \\
\hline
\end{tabular}
\end{table}

\subsection{Baselines}

We compare STAR-GS with the following baselines:

\begin{itemize}
    \item \textbf{Time-Space-Network Mixed-Integer Linear Programming (TSN-MILP).} TSN-MILP is adapted from the time-space-network formulation in Lee et al.~\cite{aerospace11010083}. We implement it with the same bid-price revenue maximization objective as Problem~\eqref{eq:opt_problem}, i.e., maximizing $\sum_{(i,k)} \mathrm{bid}_{ik}\tau_{ik}y_{ik}$ under the same scheduling constraints. It serves as an optimization benchmark for small-scale experiments.

    \item \textbf{Communication-Resource Throughput Greedy (CR-Throughput).} CR-Throughput is a communication-resource greedy baseline adapted from Ronen and Ben-Moshe~\cite{s25247538}. At each time slot, each ground station selects the currently serviceable task with the largest immediately deliverable data volume.

    \item \textbf{Earliest Deadline First (EDF).} EDF is a bid-agnostic scheduling baseline that processes tasks according to the earliest-deadline-first rule. Unlike STAR-GS, it does not use bids for admission decisions and prioritizes tasks only by deadline urgency.
    
    \item \textbf{Highest-Bid-First with Naive Assignment (HBF-NA).} 
    HBF-NA is a bid-greedy baseline inspired by auction-based resource allocation~\cite{10887300}. It admits tasks in descending bid order and assigns each task to the first serviceable ground station, without best-fit or slack-based feasibility comparison.
    
    \item \textbf{Satellite-Centric Best-Visibility (SC-BV).} 
    SC-BV is a satellite-centric baseline inspired by LEO ground-station selection studies~\cite{Duncan}. Each satellite selects the ground station with the longest total visibility duration, through which all its tasks are scheduled.
\end{itemize}

\subsection{Overall Performance}

\begin{figure*}[t]
    \centering

    \subfloat[Task volume]{
        \includegraphics[width=0.29\textwidth]{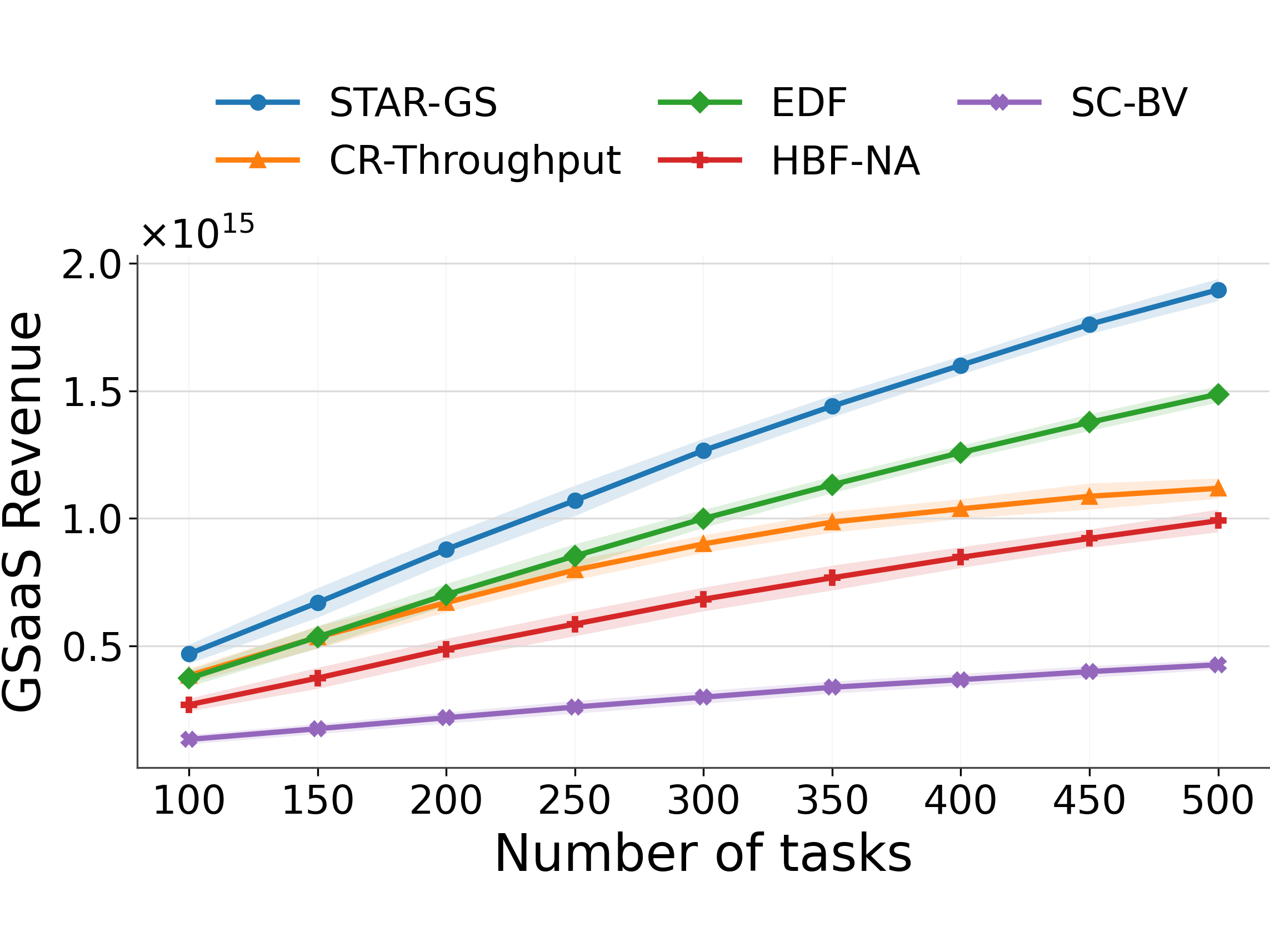}
    }
    \hfill
    \subfloat[GS bandwidth]{
        \includegraphics[width=0.29\textwidth]{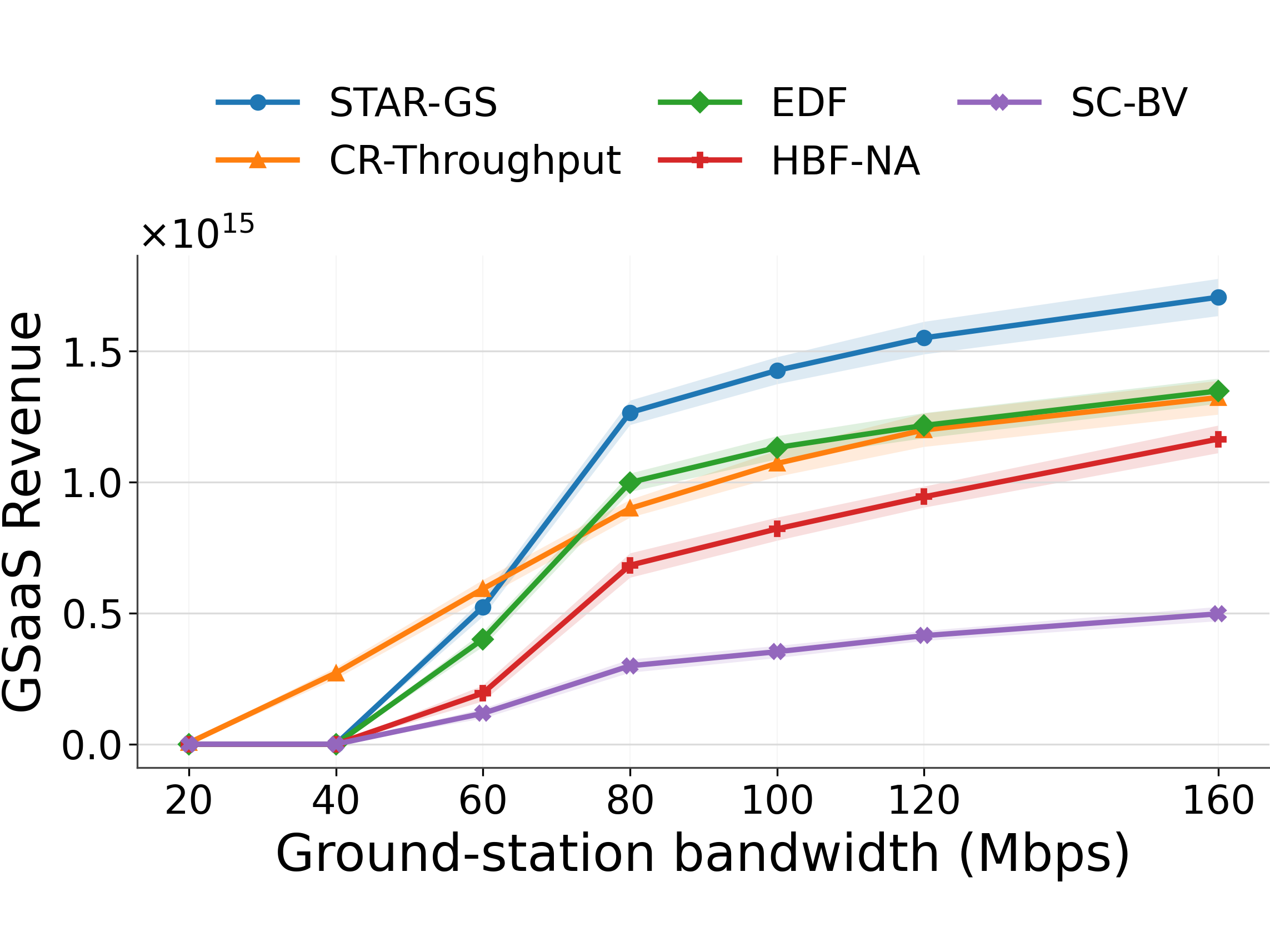}
    }
    \hfill
    \subfloat[Deadline slack]{
        \includegraphics[width=0.29\textwidth]{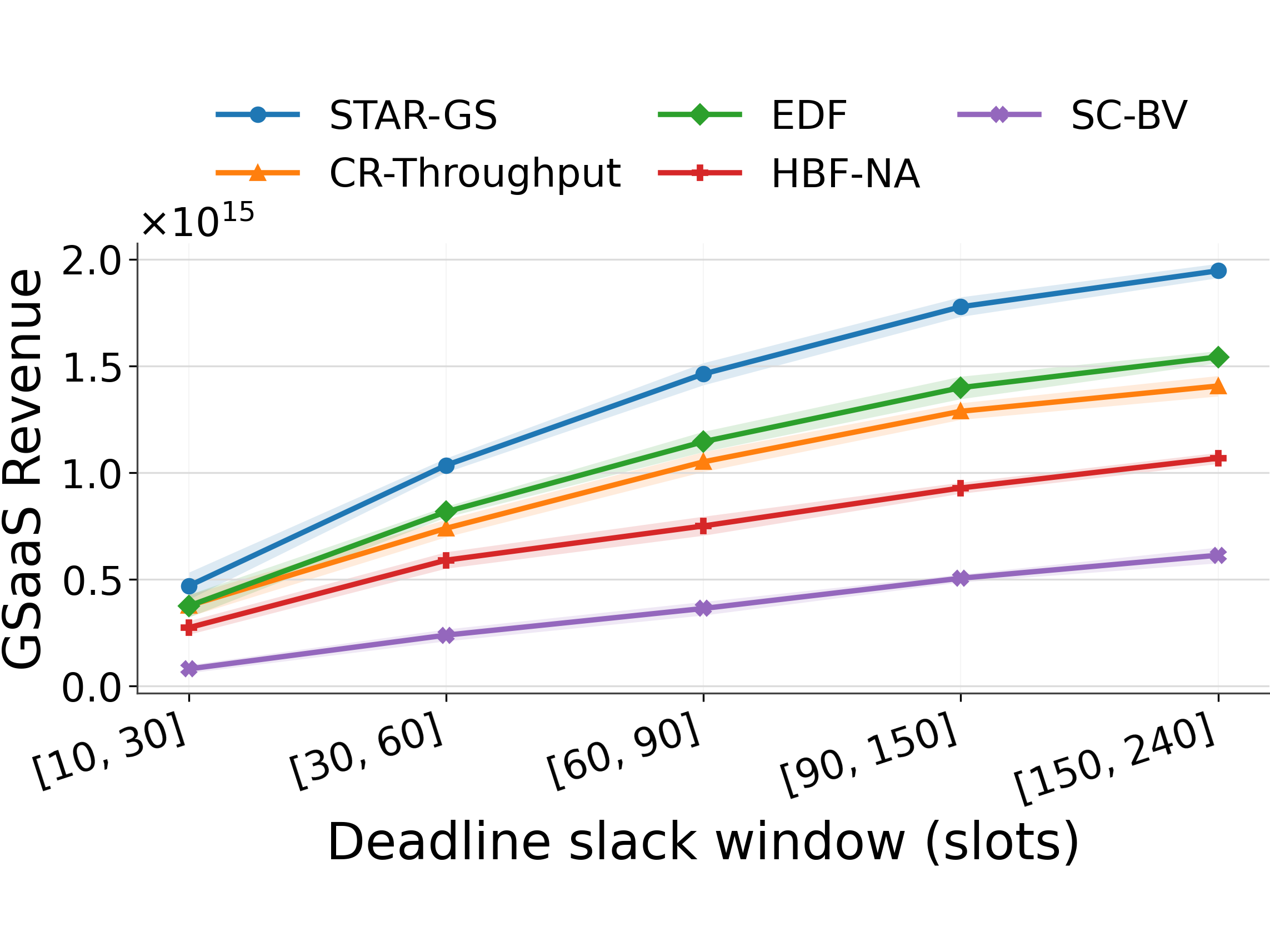}
    }

    \subfloat[Task demand]{
        \includegraphics[width=0.29\textwidth]{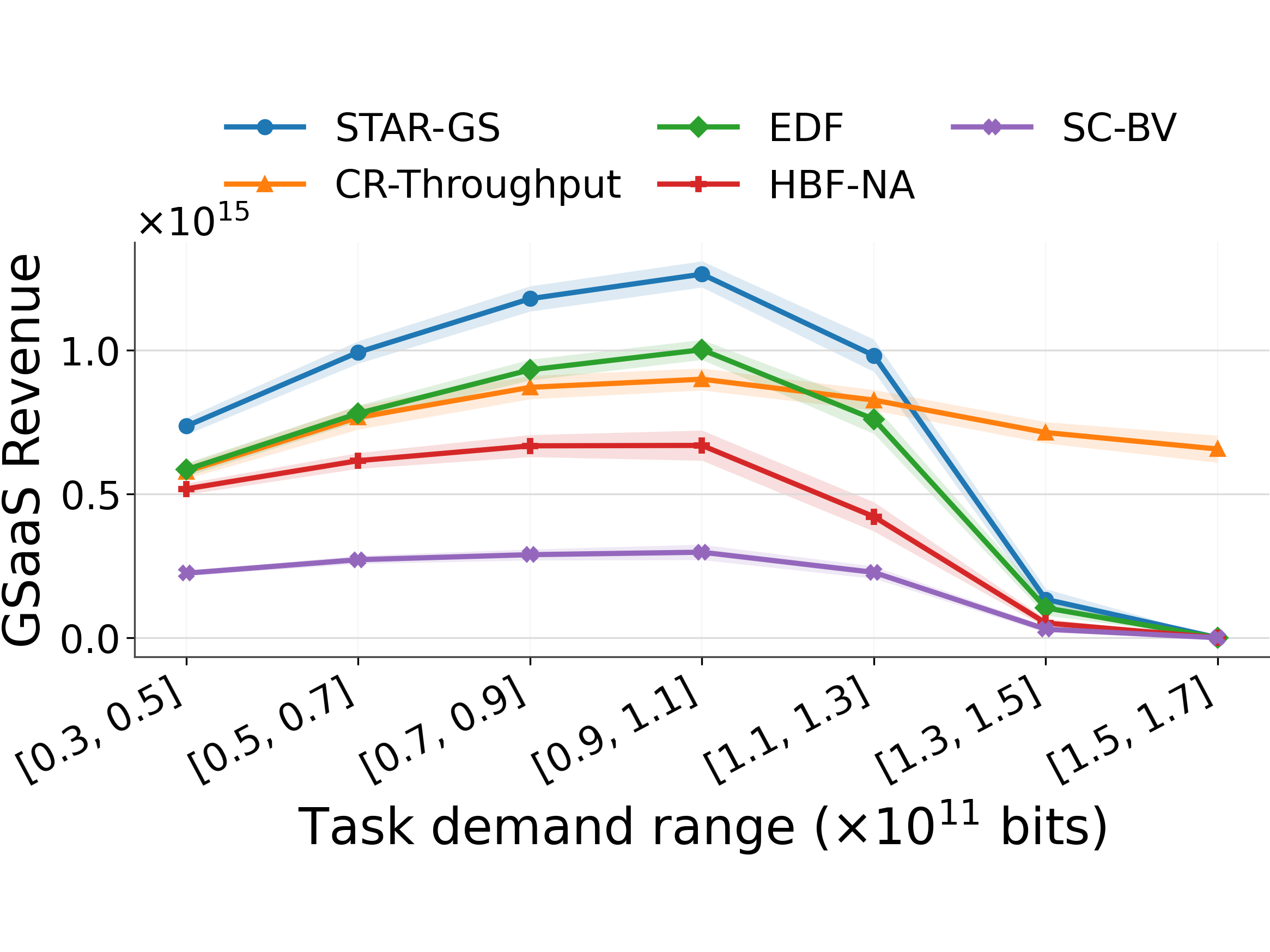}
    }
    \hfill
    \subfloat[Bid dispersion]{
        \includegraphics[width=0.29\textwidth]{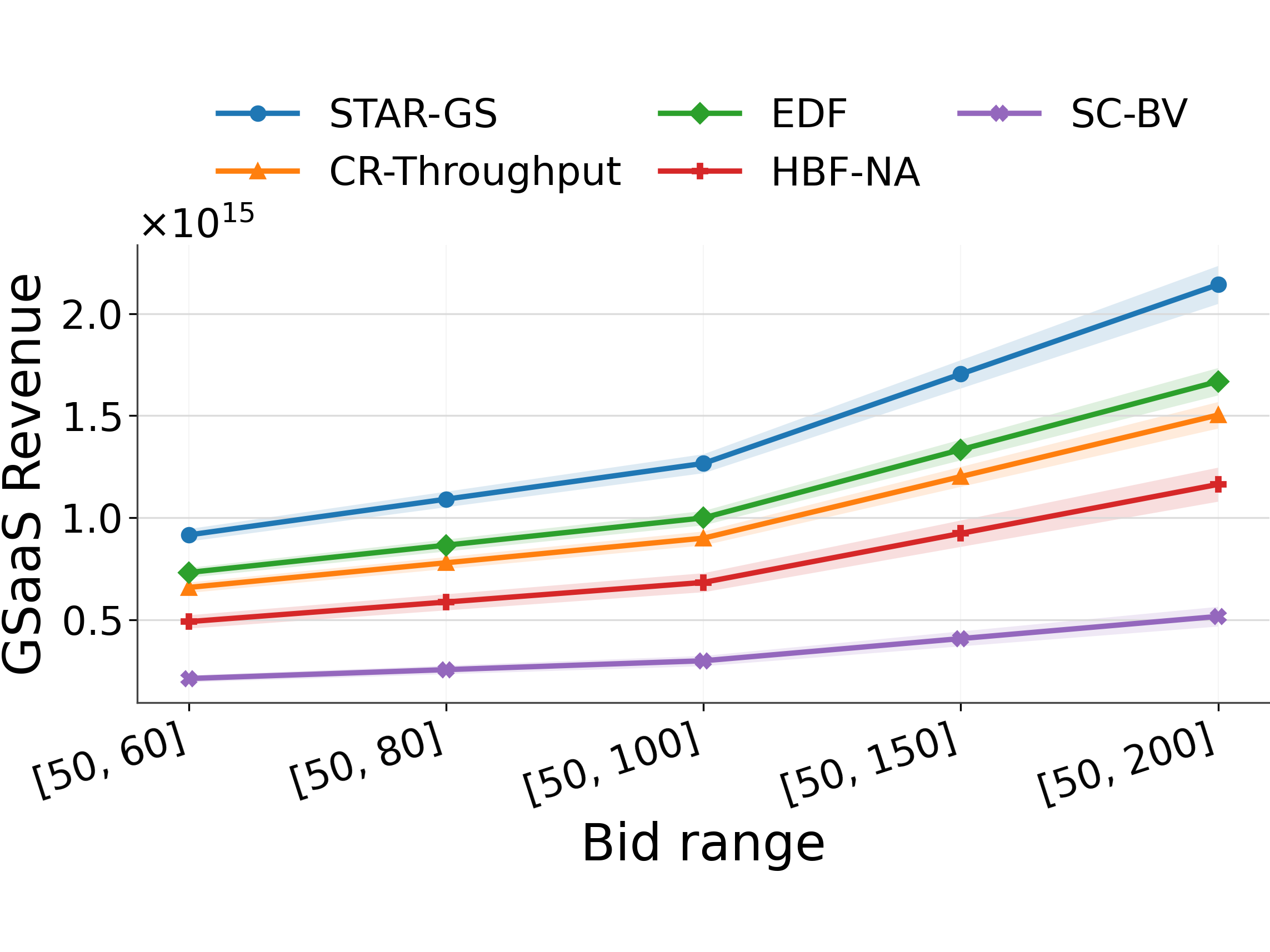}
    }
    \hfill
    \subfloat[Acceptance ratio]{
        \includegraphics[width=0.29\textwidth]{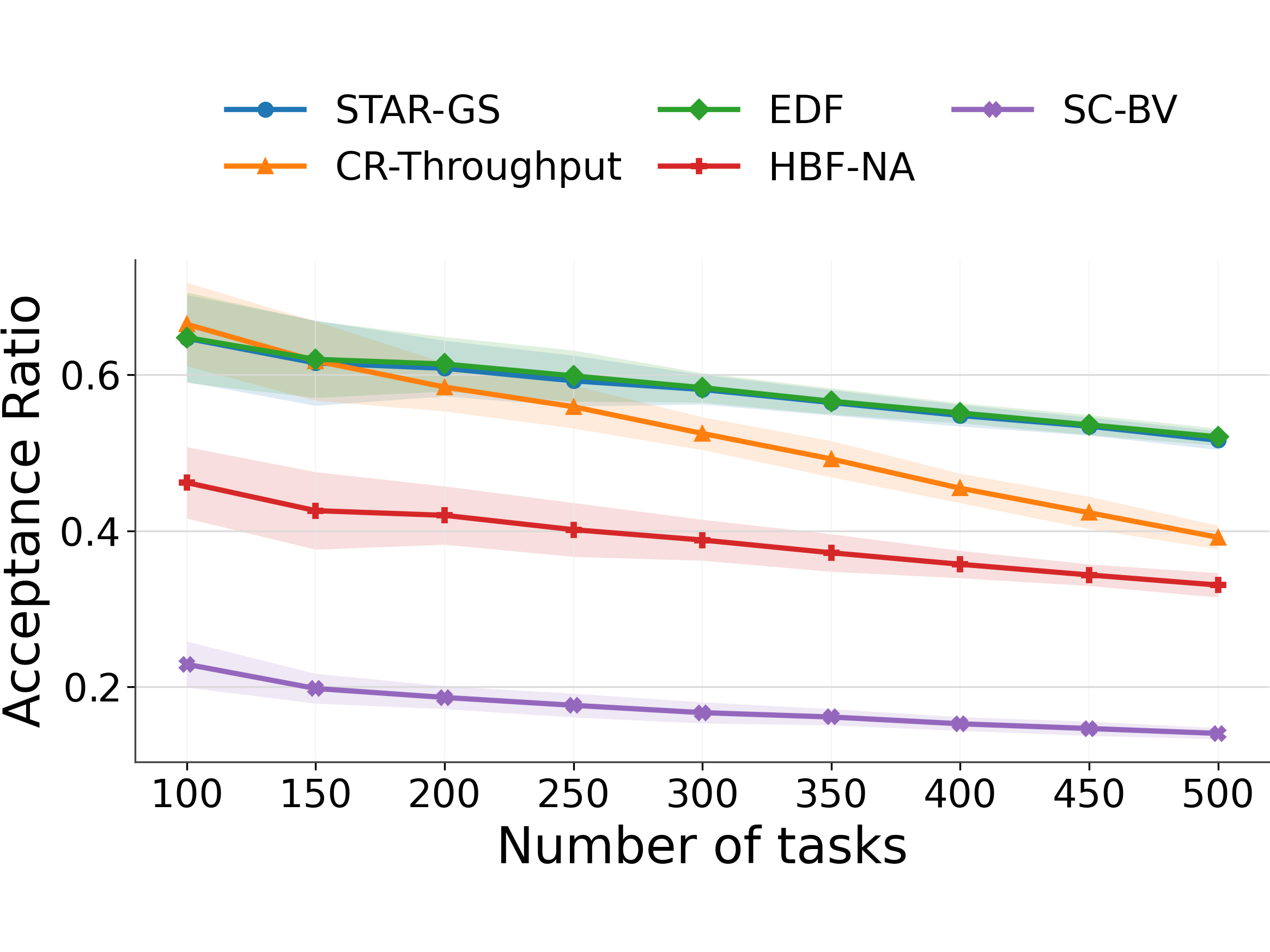}
    }

    \caption{Performance comparison under different system settings. Subfigures (a)--(e) report GSaaS revenue under different system parameters, while subfigure (f) reports the task acceptance ratio as the task volume increases.}
    \label{fig:performance}
\end{figure*}

In Fig.~\ref{fig:performance}, we evaluate STAR-GS under varying system and workload settings, including task count, ground station bandwidth, deadline slack, task demand, bid dispersion, and task acceptance ratio. Unless otherwise specified, each data point is averaged over 10 runs.

Fig.~\ref{fig:performance}(a) compares GSaaS revenue as the number of tasks increases from $100$ to $500$. STAR-GS consistently achieves the highest revenue, and its advantage becomes more pronounced under heavier workloads. At $500$ tasks, EDF reaches about $76.8\%$ of STAR-GS revenue, while HBF-NA and SC-BV achieve only about $51.0\%$ and $22.4\%$, respectively. This shows that deadline-aware scheduling alone is not sufficient for revenue maximization, and that naive GS assignment or fixed satellite-centric station selection can lead to inefficient resource usage. Overall, STAR-GS scales better by jointly considering task value, scheduling feasibility, and GS assignment.

Fig.~\ref{fig:performance}(b) shows the impact of ground station bandwidth. When bandwidth is extremely limited, most methods obtain little or no revenue because few tasks are schedulable. As bandwidth increases, all methods improve, but STAR-GS benefits more from the additional capacity and remains the best-performing method. EDF is the closest baseline in the medium- and high-bandwidth regimes, but it still lags behind because it prioritizes urgent tasks rather than high-value tasks. This confirms that additional bandwidth is more effectively utilized when admission and assignment are both value-aware and feasibility-aware.

Fig.~\ref{fig:performance}(c) evaluates the impact of deadline slack. Larger slack windows increase revenue for all methods because tasks have more feasible contact opportunities and are easier to schedule. STAR-GS consistently outperforms the baselines across different slack settings, with EDF again being the closest competitor. The gap indicates that relaxed deadlines alone do not eliminate the need for bid-aware admission, since the provider still needs to decide which feasible tasks are most valuable to serve.

Fig.~\ref{fig:performance}(d) shows the impact of task demand. Revenue first increases as task demand grows from low to moderate levels, because each admitted task contributes more potential revenue while remaining schedulable. When the demand becomes too large, however, the feasible task set shrinks quickly and revenue drops for most methods. In this high-demand regime, CR-Throughput becomes more competitive because its deliverability-oriented rule favors tasks that can still be completed under tight resource constraints.

Fig.~\ref{fig:performance}(e) shows the impact of bid dispersion. As the bid range becomes wider, all methods obtain higher revenue, but STAR-GS gains the most because it directly exploits bid heterogeneity during admission. The increasing gap between STAR-GS and bid-agnostic or throughput-oriented baselines confirms the benefit of value-aware task selection when task values differ significantly.

Fig.~\ref{fig:performance}(f) compares the task acceptance ratio as the task volume increases. All methods admit a smaller fraction of tasks under heavier workloads due to stronger competition for limited bandwidth and visibility windows. STAR-GS maintains an acceptance ratio close to EDF while achieving much higher revenue. This suggests that STAR-GS does not improve revenue simply by rejecting more tasks; instead, it selects feasible tasks with higher value more effectively.

\begin{figure}[t]
    \centering
    \subfloat[GSaaS revenue]{\includegraphics[width=0.49\linewidth]{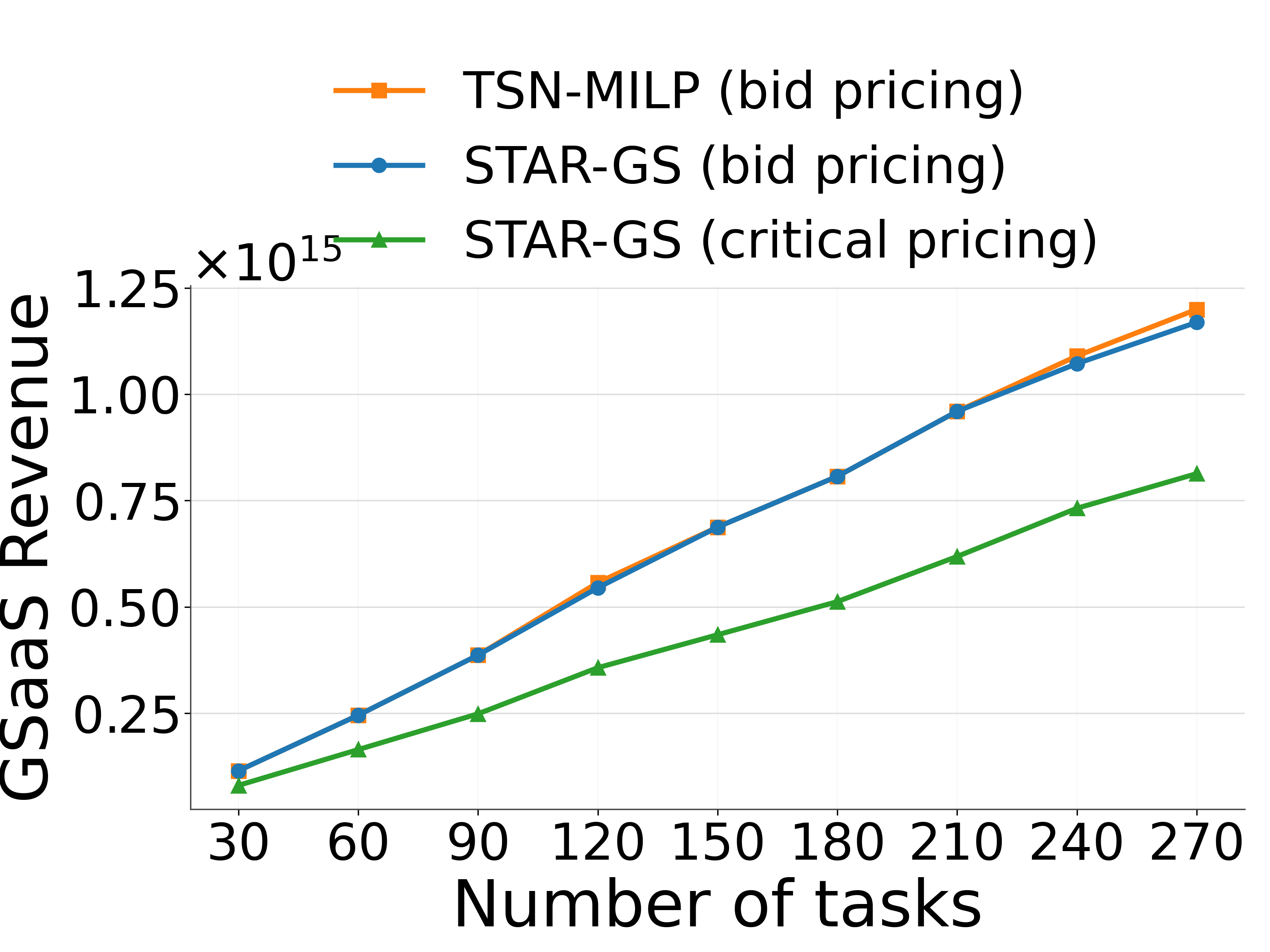}\label{fig:milp-objective}}
    \hfill
    \subfloat[Runtime]{\includegraphics[width=0.45\linewidth]{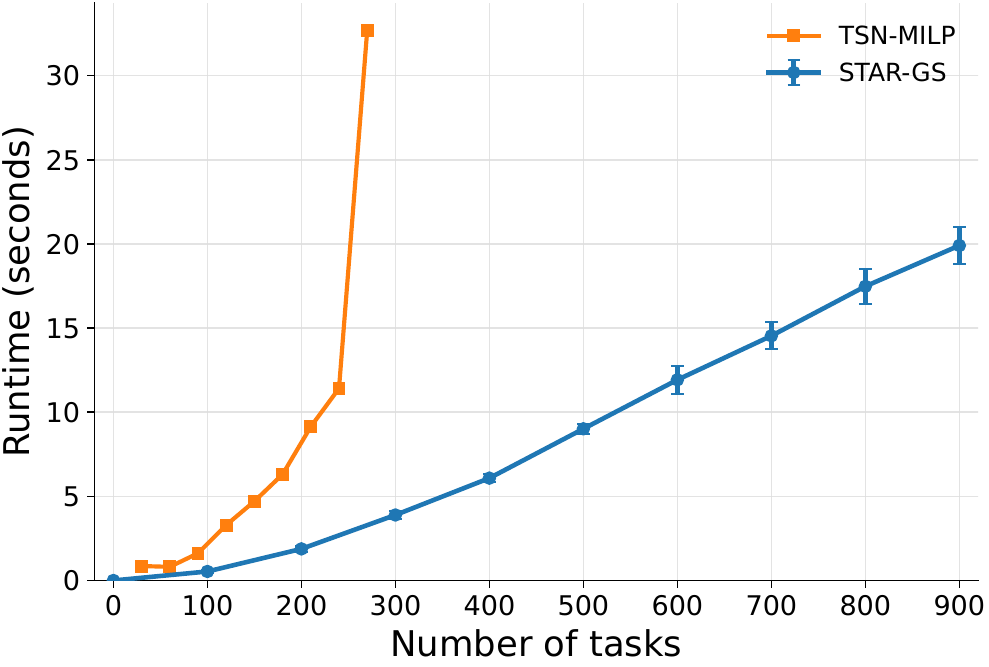}\label{fig:milp-runtime}}
    \caption{Comparison between STAR-GS and TSN-MILP. (a) GSaaS revenue under $30$--$270$ tasks. (b) Runtime of TSN-MILP under up to $270$ tasks and STAR-GS under up to $900$ tasks.}
    \label{fig:milp-comparison}
\end{figure}

We further compare STAR-GS with TSN-MILP to evaluate solution quality and computational efficiency against an optimization benchmark. Since TSN-MILP has substantially higher computational cost, we evaluate it on instances containing $30$ to $270$ tasks. To further examine the scalability of STAR-GS, we extend its runtime evaluation to workloads containing up to $900$ tasks. TSN-MILP jointly optimizes task admission, ground-station assignment, and bandwidth allocation using a mixed-integer formulation.

Fig.~\ref{fig:milp-comparison}(a) compares the GSaaS revenue achieved by STAR-GS and TSN-MILP under different numbers of tasks. The TSN-MILP and STAR-GS curves with bid pricing report the revenue calculated from submitted bids, while the STAR-GS curve with critical pricing reports the truthful payment computed by the proposed truthful mechanism. As the number of tasks increases from $30$ to $270$, the revenue of all methods increases steadily. Under bid pricing, STAR-GS achieves revenue very close to TSN-MILP, indicating that STAR-GS can closely approximate the MILP solution in task admission and resource allocation. In contrast, STAR-GS with critical pricing obtains lower revenue because the critical price is the payment required to guarantee truthfulness, rather than the submitted bid.

Fig.~\ref{fig:milp-comparison}(b) compares their runtime and further evaluates the workload scalability of STAR-GS. The runtime of TSN-MILP increases rapidly, from less than $1$ second at $30$ tasks to more than $32$ seconds at $270$ tasks. In contrast, STAR-GS exhibits a substantially slower and smoother increase in runtime. Its runtime remains below $4$ seconds at approximately $300$ tasks and reaches only about $20$ seconds at $900$ tasks. These results show that STAR-GS achieves near-MILP revenue while requiring substantially less computation time, making it more suitable for large-scale GSaaS scheduling.

Additional evaluations of operational robustness and economic efficiency are presented in Appendices~\ref{app:robustness} and~\ref{app:economic_utility}, respectively.

\section{Conclusion}
\label{sec:conclusion}

This paper presented STAR-GS, a truthful and visibility-aware scheduling mechanism for revenue-driven GSaaS resource allocation. By combining bid-aware admission, best-fit ground station assignment, EDF-based schedulability checking, and critical payment pricing, STAR-GS jointly addresses feasibility and strategic bidding under dynamic satellite--ground visibility constraints. Experiments with STK-generated contact windows show that STAR-GS consistently improves provider revenue over heuristic baselines; for example, under 500 tasks, the closest baseline achieves only about 76.8\% of STAR-GS revenue, while naive bid-based and satellite-centric baselines achieve substantially less. STAR-GS also obtains revenue close to the MILP benchmark while reducing runtime significantly, e.g., about 5 seconds versus over 32 seconds at 270 tasks. These results suggest that provider-side GSaaS scheduling benefits from treating pricing and deadline-aware feasibility as coupled decisions. Future work will extend STAR-GS to online request arrivals, time-varying channel conditions, dynamic ground-station capacity, adaptive reserve pricing, and multi-ground-station task splitting.

\appendices

\section{Robustness under Operational Stress}
\label{app:robustness}

We further evaluate STAR-GS under two stressed operational settings. First, in the heterogeneous-capacity setting, different GSs are assigned unequal bandwidth capacities while maintaining the same average capacity of $80$ Mbps as in the homogeneous setting. Second, in the bursty-arrival setting, $70\%$ of the tasks are concentrated around three burst periods, while the remaining $30\%$ are uniformly distributed across the planning horizon. All other parameters retain their default values, and each result is averaged over $10$ independent random seeds.

\begin{figure*}[t]
    \centering
    \subfloat[Heterogeneous GS capacities.]{\includegraphics[width=0.3\textwidth]{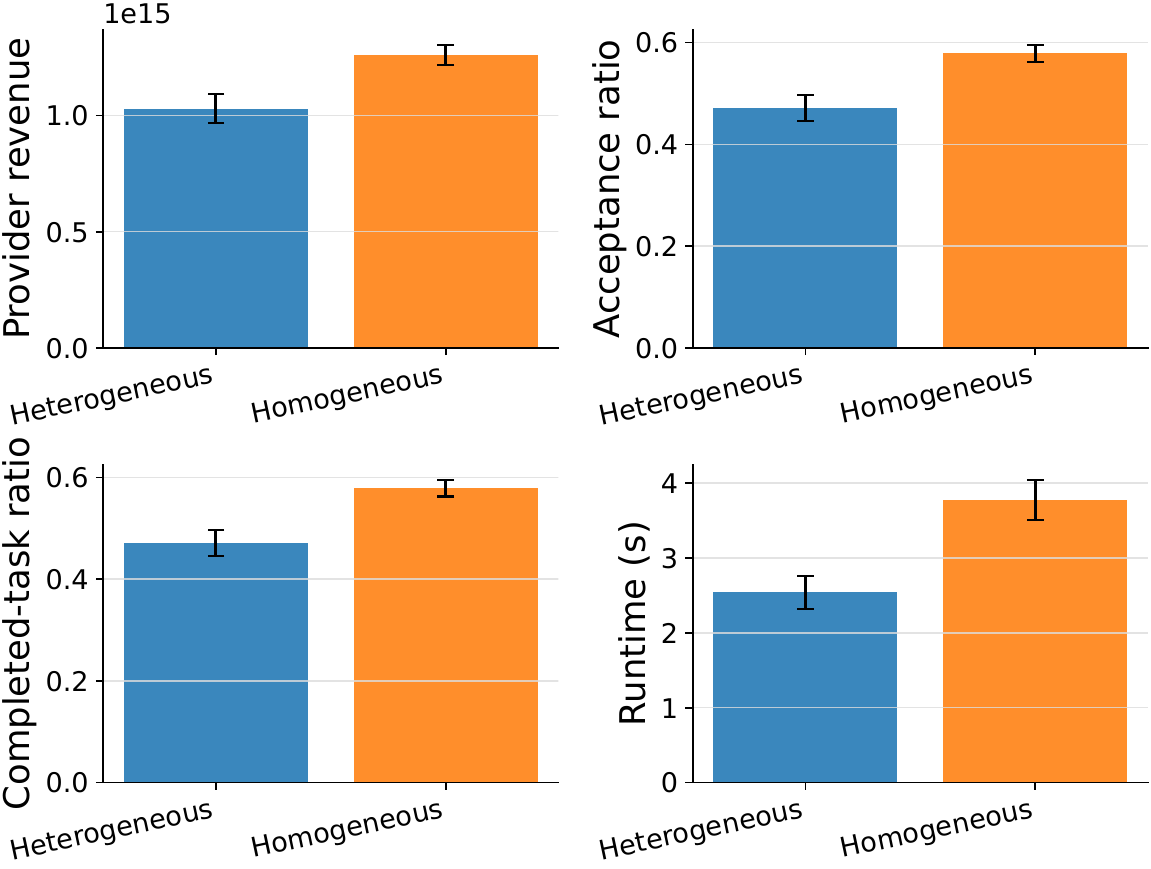}}
    \hfill
    \subfloat[Bursty task arrivals.]{\includegraphics[width=0.3\textwidth]{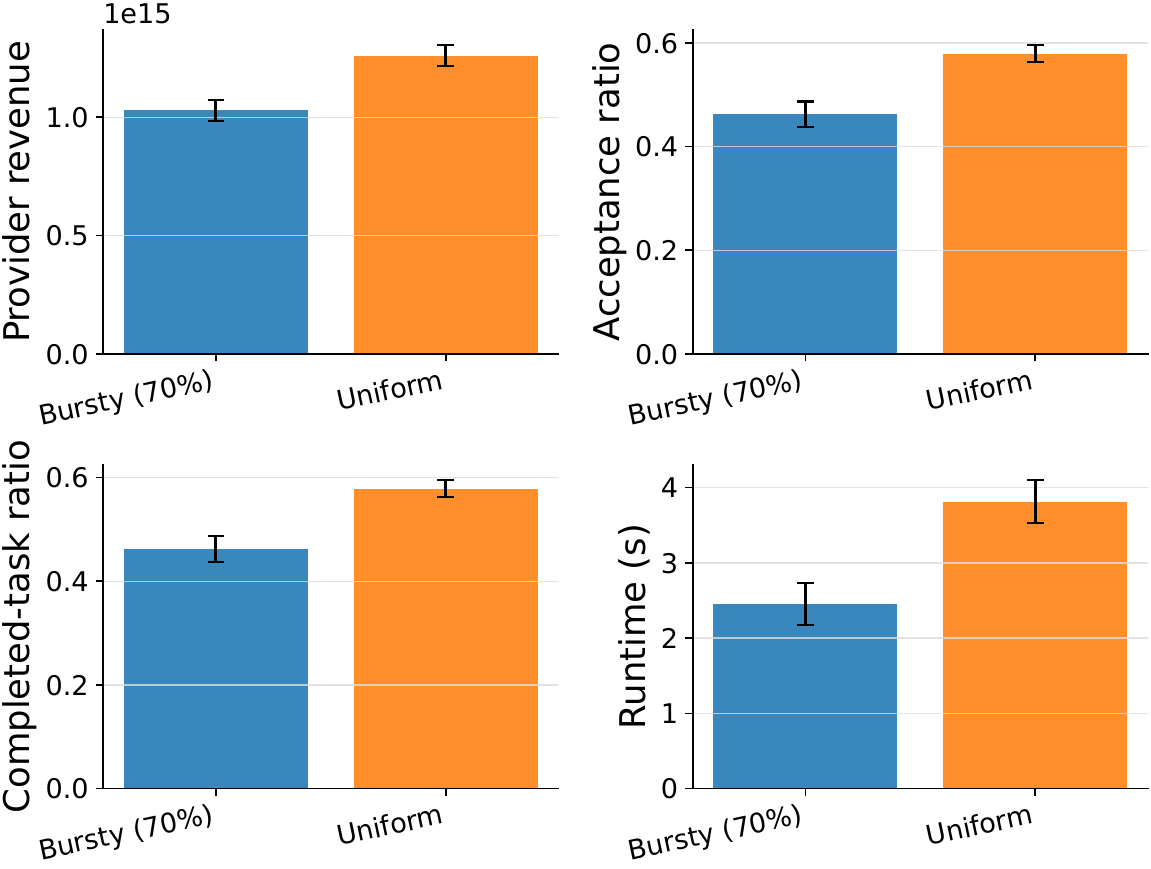}}
    \hfill
    \subfloat[Task-arrival distributions.]{\includegraphics[width=0.3\textwidth]{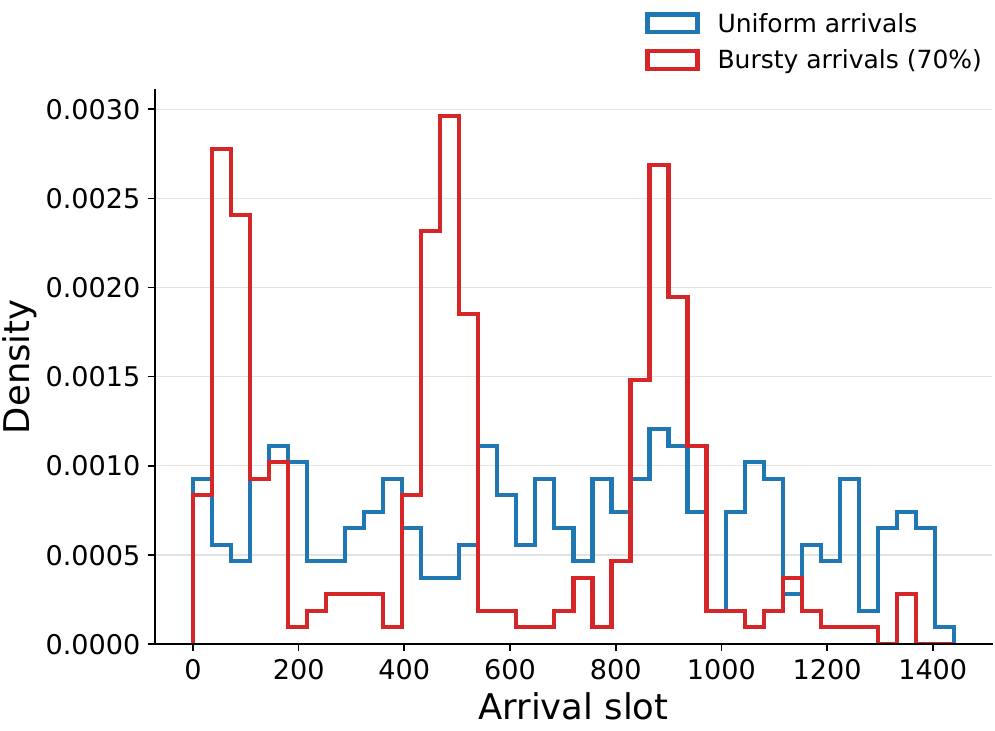}}
    \caption{Robustness of STAR-GS under operational stress. Subfigure (a) compares heterogeneous and homogeneous GS capacities, subfigure (b) compares bursty and uniform task arrivals, and subfigure (c) illustrates the task-arrival distributions of one representative paired instance. In subfigures (a) and (b), each bar reports the mean over $10$ independent random seeds, and the error bars denote $\pm 1$ population standard deviation.}
    \label{fig:robustness}
\end{figure*}

As shown in Fig.~\ref{fig:robustness}(a), heterogeneous GS capacities reduce provider revenue, task acceptance ratio, and completed-task ratio relative to the homogeneous setting. Although the two settings have the same average capacity, capacity imbalance makes some GSs resource bottlenecks while leaving capacity at other GSs less effectively utilized because satellite--ground visibility is location dependent. Nevertheless, STAR-GS continues to produce feasible schedules under heterogeneous capacities. Its runtime is lower in this setting because fewer tasks pass the feasibility-aware admission process.

Fig.~\ref{fig:robustness}(c) confirms that the bursty workload exhibits three pronounced arrival peaks, in contrast to the approximately uniform distribution of the default workload. As shown in Fig.~\ref{fig:robustness}(b), this temporal concentration creates stronger contention for overlapping visibility windows and GS bandwidth. Consequently, provider revenue, task acceptance ratio, and completed-task ratio decrease relative to the uniform-arrival setting. STAR-GS nevertheless preserves scheduling feasibility by rejecting tasks that cannot be completed before their deadlines. The lower runtime mainly results from the smaller number of feasible and admitted tasks.

\section{Economic Efficiency and User Utility}
\label{app:economic_utility}

Beyond provider revenue, we further examine how the economic value generated by completed tasks is shared between the GSaaS provider and task owners. We report the gross completed value, the provider revenue collected through critical payments, and the aggregate utility of winning task owners defined in Eq.~\eqref{user_utility}.

\begin{figure}[t]
    \centering
    \includegraphics[width=0.55\columnwidth]{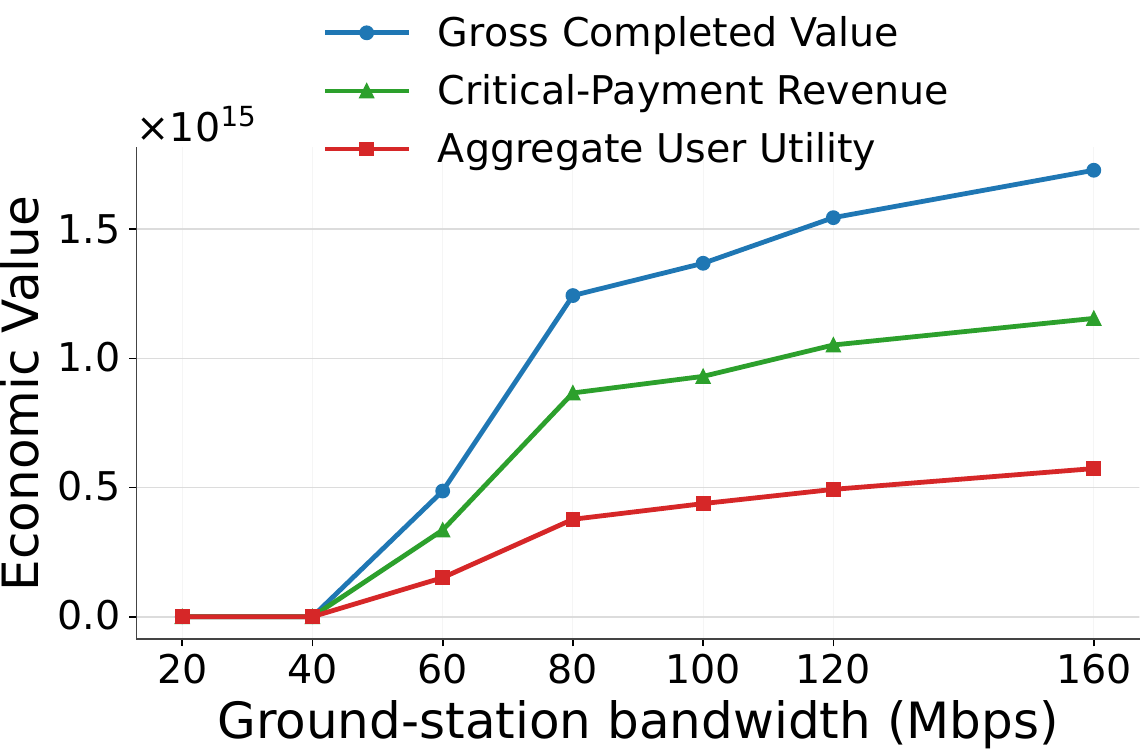}
    \caption{Economic-value decomposition under different ground-station bandwidths.}
    \label{fig:economic_decomposition}
\end{figure}

As shown in Fig.~\ref{fig:economic_decomposition}, little economic value is generated when the GS bandwidth is below $40$ Mbps because the scarcity of communication resources allows only a few tasks to be completed. As the bandwidth increases, more tasks become schedulable, leading to simultaneous increases in provider revenue and aggregate user utility. The figure also shows that critical-payment revenue and aggregate user utility sum to the gross completed value, consistent with our economic model. This illustrates how the value generated by completed tasks is shared between the provider and task owners. Therefore, additional GS capacity benefits both sides of the market.

\bibliographystyle{IEEEtran}
\bibliography{ref}

@IEEEtranBSTCTL{IEEEexample:BSTcontrol,
    CTLuse_forced_etal       = "yes",
    CTLmax_names_forced_etal = "3",
    CTLnames_show_etal       = "3" 
}

@INPROCEEDINGS{10858567,
  author={Zhao, Heng and Cen, Sheng and Zhu, Yifei},
  booktitle={2024 IEEE 32nd International Conference on Network Protocols}, 
  title={The Space Above the Sky: Uniting Global-Scale Ground Station as a Service for Efficient Orbital Data Processing}, 
  year={2024},
  volume={},
  number={},
  pages={1-11},
  doi={10.1109/ICNP61940.2024.10858567}}

@INPROCEEDINGS{11044527,
  author={Chou, Yi Ching and Chen, Long and Wang, Hengzhi and Wang, Feng and Fang, Hao and Zhao, Haoyuan and Zhang, Miao and Fan, Xiaoyi},
  booktitle={IEEE INFOCOM 2025 - IEEE Conference on Computer Communications}, 
  title={Commercial Dishes Can Be My Ladder: Sustainable and Collaborative Data Offloading in LEO Satellite Networks}, 
  year={2025},
  volume={},
  number={},
  pages={1-10},
  doi={10.1109/INFOCOM55648.2025.11044527}}

@INPROCEEDINGS{11044563,
  author={Gu, Chenwei and Wu, Qian and Lai, Zeqi and Li, Hewu and Weng, Yuxuan and Liu, Weisen and Li, Jihao and Liu, Jun and Li, Yuanjie},
  booktitle={IEEE INFOCOM 2025 - IEEE Conference on Computer Communications}, 
  title={NovaPlan: An Efficient Plan of Renting Ground Stations for Emerging LEO Satellite Networks}, 
  year={2025},
  volume={},
  number={},
  pages={1-10},
  doi={10.1109/INFOCOM55648.2025.11044563}}

@article{s25247538,
author = {Ronen, Rony and Ben-Moshe, Boaz},
year = {2025},
month = {12},
pages = {1-24},
title = {Maximizing Nanosatellite Throughput via Dynamic Scheduling and Distributed Ground Stations},
volume = {25},
journal = {Sensors},
doi = {10.3390/s25247538}
}

@article{aerospace11010083,
  title={Mixed-Integer Linear Programming Model for Scheduling Missions and Communications of Multiple Satellites},
  author={Min-Woo Lee and Seung-Duk Yu and Kybeom Kwon and Myung Bok Lee and Junghyun Lee and Heungseob Kim},
  journal={Aerospace},
  year={2024},
  pages={1-32}
}

@INPROCEEDINGS{10901122,
  author={Bhandari, Sovit and Vu, Thang X. and Chatzinotas, Symeon},
  booktitle={GLOBECOM 2024 - 2024 IEEE Global Communications Conference}, 
  title={LEO Satellite-assisted Task Offloading for a Near Real-time Earth Observation Service}, 
  
  volume={},
  number={},
  pages={3033-3038},
  year={2024}}

@article{myung2008single,
author = {Myung, Hyung and Goodman, David},
year = {2008},
month = {10},
pages = {1-185},
title = {Single Carrier FDMA: A New Air Interface for Long Term Evolution},
isbn = {9780470724491},
doi = {10.1002/9780470758717}
}

@article{a2fd2512-18e4-3e86-a4f2-2f4443994eb2,
  title={Counterspeculation, Auctions, And Competitive Sealed Tenders},
  author={William S. Vickrey},
  journal={Journal of Finance},
  year={1961},
  volume={16},
  pages={8-37}
}

@article{doi:10.1137/0203025,
author = {Johnson, D. S. and Demers, A. and Ullman, J. D. and Garey, M. R. and Graham, R. L.},
title = {Worst-Case Performance Bounds for Simple One-Dimensional Packing Algorithms},
year = {1974},
publisher = {Society for Industrial and Applied Mathematics},
address = {USA},
volume = {3},
number = {4},
issn = {0097-5397},
doi = {10.1137/0203025},
journal = {SIAM Journal on Computing},
month = dec,
pages = {299–325},
numpages = {27}
}

@article{myerson1981optimal,
  title={Optimal auction design},
  author={Myerson, Roger B},
  journal={Mathematics of operations research},
  volume={6},
  number={1},
  pages={58--73},
  year={1981},
  publisher={INFORMS}
}

@article{nisan2007algorithmic,
author = {Roughgarden, Tim},
title = {Algorithmic game theory},
year = {2010},
issue_date = {July 2010},
publisher = {Association for Computing Machinery},
address = {New York, NY, USA},
volume = {53},
number = {7},
issn = {0001-0782},
doi = {10.1145/1785414.1785439},
journal = {Commun. ACM},
pages = {78--86},
numpages = {9}
}

@ARTICLE{11329400,
  author={Ren, Yingying and Li, Qiuli and Zhang, Yue and Li, Jiting and Pedrycz, Witold and Nagaratnam Suganthan, Ponnuthurai and Mallipeddi, Rammohan and Song, Yanjie},
  journal={IEEE Transactions on Systems, Man, and Cybernetics: Systems}, 
  title={An Evolutionary Algorithm With Memory Guidance for Data Transmission Scheduling Optimization in Communication Satellite Network}, 
  year={2026},
  volume={56},
  number={2},
  pages={1354-1368},
  doi={10.1109/TSMC.2025.3647265}}

@INPROCEEDINGS{10115903,
  author={Boschetti, Nicolò and Smethurst, Chelsea and Epiphaniou, Gregory and Maple, Carsten and Sigholm, Johan and Falco, Gregory},
  booktitle={2023 IEEE Aerospace Conference}, 
  title={Ground Station as a Service Reference Architectures and Cyber Security Attack Tree Analysis}, 
  year={2023},
  volume={},
  number={},
  pages={1-12},
  doi={10.1109/AERO55745.2023.10115903}}

@INPROCEEDINGS{10154485,
  author={Filho, Edilson and Silveira, Jarbas and Marcon, César},
  booktitle={2023 IEEE 24th Latin American Test Symposium}, 
  title={Blockchain Applied In Decentralization of Ground Stations To Educational Nanosatellites}, 
  year={2023},
  volume={},
  number={},
  pages={1-5},
  doi={10.1109/LATS58125.2023.10154485}}

@article{kim2025scalablegroundstationselection,
  title={Scalable Ground Station Selection for Large LEO Constellations},
  author={Kim, Grace Ra and Eddy, Duncan and Srinivas, Vedant and Kochenderfer, Mykel J},
  journal={arXiv preprint arXiv:2510.03438},
  year={2025}
}

@INPROCEEDINGS{10001260,
  author={Velusamy, Gandhimathi and Lent, Ricardo},
  booktitle={GLOBECOM 2022 - 2022 IEEE Global Communications Conference}, 
  title={AI-Based Ground Station-as-a-Service for Optimal Cost-Latency Satellite Data Downloading}, 
  year={2022},
  volume={},
  number={},
  pages={2363-2368},
  doi={10.1109/GLOBECOM48099.2022.10001260}}

@article{CHENG2026123200,
author = {Cheng, Jiaqi and Fan, Ming and Yi, Gu and Tang, Wei and Luo, Qizhang and Wang, Yalin and Wang, Xinwei and Wu, Guohua},
year = {2026},
pages = {1-19},
title = {UHTS-DRL: A deep reinforcement learning framework for integrated agile satellite observation and data transmission scheduling},
volume = {740},
journal = {Information Sciences},
doi = {10.1016/j.ins.2026.123200}
}

@inproceedings{10.1145/3680207.3765249,
author = {Sun, Zehua and Ni, Tao and Hu, Pengfei and Gu, Tao and Xu, Weitao},
title = {SpaceSched: A Constellation-Wide Scheduling System for Resolving Ground Track Congestion in Remote Sensing},
year = {2025},
isbn = {9798400711299},
doi = {10.1145/3680207.3765249},
booktitle = {Proceedings of the 31st Annual International Conference on Mobile Computing and Networking},
pages = {832–847},
numpages = {16},
location = {Kerry Hotel, Hong Kong, Hong Kong, China}
}

@article{REN2025126303,
title = {A distance similarity-based genetic optimization algorithm for satellite ground network planning considering feeding mode},
journal = {Expert Systems with Applications},
volume = {268},
pages = {1-17},
year = {2025},
issn = {0957-4174},
author = {Yingying Ren and Qiuli Li and Yangyang Guo and Witold Pedrycz and Lining Xing and Anfeng Liu and Yanjie Song},
}

@ARTICLE{10979921,
  author={Yu, Jiguo and Liu, Shun and Zou, Yifei and Wang, Guijuan and Hu, Chunqiang},
  journal={IEEE Internet of Things Journal}, 
  title={Auction Theory and Game Theory Based Pricing of Edge Computing Resources: A Survey}, 
  year={2025},
  volume={12},
  number={16},
  pages={32394-32418},
  doi={10.1109/JIOT.2025.3565539}}

@ARTICLE{10887300,
  author={Wang, Xueyi and Wang, Xingwei and Wang, Chen and Zeng, Rongfei and Ma, Lianbo and He, Qiang and Huang, Min},
  journal={IEEE Transactions on Mobile Computing}, 
  title={Truthful Online Combinatorial Auction-Based Mechanisms for Task Offloading in Mobile Edge Computing}, 
  year={2025},
  volume={24},
  number={7},
  pages={6488-6502},
  doi={10.1109/TMC.2025.3542135}}

@ARTICLE{929853,
  author={Akyildiz, I.F. and Morabito, G. and Palazzo, S.},
  journal={IEEE/ACM Transactions on Networking}, 
  title={TCP-Peach: a new congestion control scheme for satellite IP networks}, 
  year={2001},
  volume={9},
  number={3},
  pages={307-321},
  doi={10.1109/90.929853}}

@ARTICLE{10520815,
  author={Yang, Wenjun and Cai, Lin and Shu, Shengjie and Pan, Jianping},
  journal={IEEE Transactions on Mobile Computing}, 
  title={Mobility-Aware Congestion Control for Multipath QUIC in Integrated Terrestrial Satellite Networks}, 
  year={2024},
  volume={23},
  number={12},
  pages={11620-11634},
  doi={10.1109/TMC.2024.3397164}}

@ARTICLE{6094268,
  author={Shah, Devavrat and Tse, David N. C. and Tsitsiklis, John N.},
  journal={IEEE Transactions on Information Theory}, 
  title={Hardness of Low Delay Network Scheduling}, 
  year={2011},
  volume={57},
  number={12},
  pages={7810-7817},
  doi={10.1109/TIT.2011.2168897}}

@misc{aws_contact,
  author       = {{Amazon Web Services}},
  title        = {{AWS Ground Station User Guide: Work with Contacts}},
  howpublished = {[Online]. Available: \url{https://docs.aws.amazon.com/ground-station/latest/ug/contacts.html}},
  note         = {Accessed: May 21, 2026}
}

@misc{agi_stk,
  author       = {{Analytical Graphics, Inc.}},
  title        = {{Systems Tool Kit (STK)}},
  howpublished = {[Online]. Available: \url{https://www.agi.com/products/stk}},
  note         = {Accessed: May 21, 2026}
}

@ARTICLE{10945753,
  author={Hui, Mingming and Zhai, Shenghua and Wang, Daqing and Hui, Tengfei and Wang, Wei and Du, Panpan and Gong, Fengkui},
  journal={IEEE Internet of Things Journal}, 
  title={A Review of LEO-Satellite Communication Payloads for Integrated Communication, Navigation, and Remote Sensing: Opportunities, Challenges, Future Directions}, 
  year={2025},
  volume={12},
  number={12},
  pages={18954-18992},
  doi={10.1109/JIOT.2025.3553942}}

@ARTICLE{9442378,
  author={Centenaro, Marco and Costa, Cristina E. and Granelli, Fabrizio and Sacchi, Claudio and Vangelista, Lorenzo},
  journal={IEEE Communications Surveys \& Tutorials}, 
  title={A Survey on Technologies, Standards and Open Challenges in Satellite IoT}, 
  year={2021},
  volume={23},
  number={3},
  pages={1693-1720},
  doi={10.1109/COMST.2021.3078433}}

@ARTICLE{9840374,
  author={Prol, F. S. and Ferre, R. Morales and Saleem, Z. and Välisuo, P. and Pinell, C. and Lohan, E. S. and Elsanhoury, M. and Elmusrati, M. and Islam, S. and Çelikbilek, K. and Selvan, K. and Yliaho, J. and Rutledge, K. and Ojala, A. and Ferranti, L. and Praks, J. and Bhuiyan, M. Z. H. and Kaasalainen, S. and Kuusniemi, H.},
  journal={IEEE Access}, 
  title={Position, Navigation, and Timing (PNT) Through Low Earth Orbit (LEO) Satellites: A Survey on Current Status, Challenges, and Opportunities}, 
  year={2022},
  volume={10},
  number={},
  pages={83971-84002},
  doi={10.1109/ACCESS.2022.3194050}}

@article{osoro2021techno,
  title={A techno-economic framework for satellite networks applied to low earth orbit constellations: Assessing Starlink, OneWeb and Kuiper},
  author={Osoro, Ogutu B and Oughton, Edward J},
  journal={IEEE Access},
  volume={9},
  pages={141611--141625},
  year={2021},
  publisher={IEEE}
}

@misc{planet_tasking_api,
  author       = {{Planet Labs PBC}},
  title        = {{Tasking API}},
  howpublished = {[Online]. Available: \url{https://docs.planet.com/develop/apis/tasking/}},
  note         = {Accessed: May 21, 2026}
}

@misc{skyfi_tasking,
  author       = {{SkyFi}},
  title        = {{Satellite Tasking Explained: How to Order Custom Satellite Imagery}},
  howpublished = {[Online]. Available: \url{https://skyfi.com/en/blog/satellite-tasking-explained}},
  note         = {Accessed: May 21, 2026}
}

@misc{mathworks_walker,
  author       = {{MathWorks}},
  title        = {{Constellation Modeling with the Orbit Propagator Block}},
  howpublished = {[Online]. Available: \url{https://www.mathworks.com/help/aeroblks/constellation-modeling-with-the-orbit-propagator-block.html}},
  note         = {Accessed: May 22, 2026}
}

@ARTICLE{363d0a634b234a96a2f672e30e105fe3,
  author={Alhusenat, Ahmad Y. and Tian, Jinjin and Rababah, Hana and Wang, Qizhou and You, Lei and Zhang, Xingjun and Lei, Lei and Chatzinotas, Symeon},
  journal={IEEE Open Journal of Vehicular Technology}, 
  title={Towards Sustainable LEO Satellite Operations: A Degradation-Aware Scheduling Framework for DoD Control and Energy Optimization}, 
  year={2026},
  volume={},
  number={},
  pages={1-16},
  doi={10.1109/OJVT.2026.3693843}}

@INPROCEEDINGS{Duncan,
  author={Eddy, Duncan and Ho, Michelle and Kochenderfer, Mykel J.},
  booktitle={2025 IEEE Aerospace Conference}, 
  title={Optimal Ground Station Selection for Low-Earth Orbiting Satellites}, 
  year={2025},
  volume={},
  number={},
  pages={1-13},
  doi={10.1109/AERO63441.2025.11068558}}

@INPROCEEDINGS{4927011,
  author={Izakian, Hesam and Ladani, Behrouz Tork and Zamanifar, Kamran and Abraham, Ajith and Snasel, Vaclav},
  booktitle={2009 IEEE Symposium on Computational Intelligence in Scheduling}, 
  title={A continuous double auction method for resource allocation in computational grids}, 
  year={2009},
  volume={},
  number={},
  pages={29-35},
  doi={10.1109/SCIS.2009.4927011}}

\vfill
\end{document}